\documentclass[11pt]{article}
\usepackage{amssymb}

\usepackage{graphicx}
\usepackage{amsmath}
\usepackage{makeidx}
\usepackage{indentfirst}

\newcounter{resultnum}[section]
\newtheorem{conclusion}{Conclusion}[section]

\newcounter{conclusionnum}[section]

\newcounter{conditionnum}[section]

\newcounter{conjecturenum}[section]

\newcounter{examplenum}[section]

\newcounter{exercisenum}[section]
\newtheorem{lemma}{Lemma}[section]

\newcounter{lemmanum}[section]

\newcounter{notationnum}[section]
\newtheorem{theorem}{Theorem}[section]

\newcounter{theoremnum}[section]

\newcounter{definitionnum}[section]
\newtheorem{corollary}{Corollary}[section]

\newcounter{corollarynum}[section]

\newcounter{remarknum}[section]

\newcounter{propositionnum}[section]

\newcounter{acknowledgementnum}[section]

\newcounter{algorithmnum}[section]

\newcounter{axiomnum}[section]

\newcounter{casenum}[section]

\newcounter{claimnum}[section]

\newcounter{summarynum}[section]

\newcounter{problemnum}[section]
\newenvironment{proof}[1][]{\textbf{Proof.} }{}

\begin{document}

\title{Fractional Curve Flows and Solitonic Hierarchies \\
in Gravity and Geometric Mechanics}
\date{April 7, 2011}
\author{\textbf{Dumitru Baleanu}\thanks{%
On leave of absence from Institute of Space Sciences, P. O. Box, MG-23, R
76900, Magurele--Bucharest, Romania, \newline
E--mails: dumitru@cancaya.edu.tr, baleanu@venus.nipne.ro, Phone:
+90-312-2844500} \\
\textsl{\small Department of Mathematics and Computer Sciences,} \\
\textsl{\small \c Cankaya University, 06530, Ankara, Turkey } \\
\and 
\textbf{Sergiu I. Vacaru} \thanks{%
sergiu.vacaru@uaic.ro, Sergiu.Vacaru@gmail.com \newline
http://www.scribd.com/people/view/1455460-sergiu} \and \textsl{\small %
Science Department, University "Al. I. Cuza" Ia\c si, } \\
\textsl{\small 54, Lascar Catargi street, Ia\c si, Romania, 700107 } }
\maketitle

\begin{abstract}
Methods from the geometry of nonholonomic manifolds and Lagran\-ge--Finsler spaces are applied in fractional calculus with Caputo derivatives and for elaborating models of fractional gravity and fractional Lagrange mechanics. The geometric data for such models are encoded into (fractional) bi--Hamiltonian structures and associated solitonic hierarchies. The constructions yield horizontal/vertical pairs of fractional vector
sine--Gordon equations and fractional  vector mKdV equations when the
hierarchies for corresponding curve fractional flows are described in
explicit forms by fractional wave maps and analogs of Schr\" odinger maps.

\vskip0.2cm


\vskip3pt


\end{abstract}


\section{Introduction}

The goal of this paper is to show how fractional solitonic hierarchies can
be canonically generated in various models of fractional gravity and
geometric Lagrange mechanics. Such constructions are possible in explicit
form for a class of fractional derivatives resulting in zero for actions on
constants, for instance, for the Caputo fractional derivative \cite%
{trujillo,agrawal,13,klimek,baleanu1,baleanu2}. This property is crucial for
constructing geometric models of theories with fractional calculus even,
after corresponding nonholonomic deformations, we may prefer to work with
another type of fractional derivatives.

This is the second partner work of paper \cite{bv3} (we also recommend
readers to consult in advance the papers \cite{bv1,bv2,vrfrf,vrfrg} on
details, notation conventions and bibliography) where we proved an important
result that via nonholonomic deformations on fractional manifolds and bundle
spaces, determined by a generating fundamental Lagrange/ --Finsler, or an
Einstein metric, we can construct linear connections with constant
coefficient curvature. For such fractional "covariant" connections, it is
possible to provide a formal encoding of integer and non--integer
gravitational dynamics, Ricci flow evolution and constrained
Lagrange/Hamilton mechanics into hierarchies of solitonic equations.

The most important consequence of such geometric studies is that using
bi--Hamilton models and related solitonic systems we can study analytically
and numerically, as well to try to construct some analogous mechanical
systems, with the aim to mimic a nonlinear/fractional nonholonomic
dynamics/evolution and even to provide certain schemes of quantization, like
in the "fractional" Fedosov approach \cite{bv2,vfed4}.

This work is organized in the form: In section \ref{s2}, we remember the
most important formulas on Caputo fractional derivatives and nonlinear
connections. Section \ref{s3} is devoted to definition of basic equations
for fractional curve flows. The Main Theorem on fractional bi--Hamiltonians
and solitonic hierarchies is formulated and proved in \ref{s4}. Finally, we
derive in general form the corresponding nonholonomic fractional solitonic
hierarchies in section \ref{s5}.

\section{Caputo Fractional Derivatives and N--connecti\-ons}

\label{s2}We summarize some important formulas on fractional calculus for
nonholonomic manifold elaborated both in global and coordinate free forms,
as well with important local integro--differential parametrizations,  in
Refs. \cite{vrfrf,vrfrg,bv1,bv3}. \ Readers are recommended to  study  in
advance those works (and references therein) on fractional differential
geometry and applications. \ Such a calculus is nonlocal both on space
and/or time coordinates when the algebra of fractional derivatives does not
have the same properties as in the integer case. Nevertheless, having
well--defined concepts of integral calculus for curved  nonholonomic
manifolds and bundles  of integer dimension, we can introduce such algebras
in local forms and then globalize the constructions using corresponding
charts on atlases covering corresponding spaces.

Our geometric arena consists from an abstract fractional manifold $\overset{%
\alpha }{\mathbf{V}}$ (we shall use also the term ''fractional space'' as an
equivalent one enabled with certain fundamental geometric structures) with
prescribed nonholonomic distribution modeling both the fractional calculus
and the non--integrable dynamics of interactions. We note that in our works
a corresponding system of notations is elaborated in a form to unify the
approaches on fractional calculus, nonholonomic bundle spaces, nonlinear
connection (N--connection) formalism etc. There are considered boldface
symbols, over/under and left up/low labels etc as we shall explain below.

Let us consider that  $f(x)$ is a derivable function  $f:[\ _{1}x,\
_{2}x]\rightarrow \mathbb{R},$ for $\mathbb{R}\ni \alpha >0,$ and denote the
derivative on $x$ as $\partial _{x}=\partial /\partial x.$ We note by $\
_{1}x$ and$\ _{2}x$ (with left low labels, respectively, $1$ and 2) two ends
of a real line interval. \ The fractional left, respectively, right Caputo
derivatives are denoted in the form
\begin{eqnarray}
&&\ _{\ _{1}x}\overset{\alpha }{\underline{\partial }}_{x}f(x):=\frac{1}{%
\Gamma (s-\alpha )}\int\limits_{\ \ _{1}x}^{x}(x-\ x^{\prime })^{s-\alpha
-1}\left( \frac{\partial }{\partial x^{\prime }}\right) ^{s}f(x^{\prime
})dx^{\prime };  \label{lfcd} \\
&&\ _{\ x}\overset{\alpha }{\underline{\partial }}_{\ _{2}x}f(x):=\frac{1}{%
\Gamma (s-\alpha )}\int\limits_{x}^{\ _{2}x}(x^{\prime }-x)^{s-\alpha
-1}\left( -\frac{\partial }{\partial x^{\prime }}\right) ^{s}f(x^{\prime
})dx^{\prime }\ .  \notag
\end{eqnarray}%
For instance,  we emphasize that the integral is considered from $\ _{1}x$
to $x$ in the symbol of partial derivative $\ _{\ _{1}x}\overset{\alpha }{%
\underline{\partial }}_{x}.$ We shall always put $\alpha $ over a symbol (or
up-left/-right to such a symbol) \ in order to emphasize that the
constructions are considered for a fractional calculus with  $\alpha \in
(0,1).$  There will \ be underlined some corresponding symbols if their
definition is strictly related to the concept of Caputo derivative.  \
Using  operators (\ref{lfcd}) generalized on $\mathbb{R}^{n},$ where $%
n=1,2...,$ and the same fractional $\alpha $ is associated to  any such
integer dimension,  we can construct the fractional absolute differential  $%
\overset{\alpha }{d}:=(dx^{j})^{\alpha }\ \ _{\ 0}\overset{\alpha }{%
\underline{\partial }}_{j}$ when $\ \overset{\alpha }{d}x^{j}=(dx^{j})^{%
\alpha }\frac{(x^{j})^{1-\alpha }}{\Gamma (2-\alpha )},$ where we consider $%
\ _{1}x^{i}=0.$

We denote a fractional tangent bundle in the form $\overset{\alpha }{%
\underline{T}}M$ \ for $\alpha \in (0,1),$ associated to a manifold $M$ of
necessary smooth class and integer $\dim M=n.$\footnote{%
The symbol $T$ is underlined in order to emphasize that we shall associate
the approach to a fractional Caputo derivative.} Locally, both the integer
and fractional local coordinates are written in the form $u^{\beta
}=(x^{j},y^{a}),$ where indices to coordinates run values $i,j,...=1,2,...n$
(for coordinates on base manifold) and $a,b,..=n+1,n+2,...,n+n$ (for typical
fiber coordinates) for integer dimensions but keep in mind that the local
derivatives are of fractional type (\ref{lfcd}), associated respectively to
any such integer coordinate.  A fractional frame basis $\overset{\alpha }{%
\underline{e}}_{\beta }=e_{\ \beta }^{\beta ^{\prime }}(u^{\beta })\overset{%
\alpha }{\underline{\partial }}_{\beta ^{\prime }}$ on $\overset{\alpha }{%
\underline{T}}M$ \ is connected via a vierlbein transform $e_{\ \beta
}^{\beta ^{\prime }}(u^{\beta })$ with a fractional local coordinate basis
\begin{equation}
\overset{\alpha }{\underline{\partial }}_{\beta ^{\prime }}=\left( \overset{%
\alpha }{\underline{\partial }}_{j^{\prime }}=_{\ _{1}x^{j^{\prime }}}%
\overset{\alpha }{\underline{\partial }}_{j^{\prime }},\overset{\alpha }{%
\underline{\partial }}_{b^{\prime }}=_{\ _{1}y^{b^{\prime }}}\overset{\alpha
}{\underline{\partial }}_{b^{\prime }}\right) ,  \label{frlcb}
\end{equation}%
for $j^{\prime }=1,2,...,n$ and $b^{\prime }=n+1,n+2,...,n+n.$ The
fractional co--bases are written $\overset{\alpha }{\underline{e}}^{\ \beta
}=e_{\beta ^{\prime }\ }^{\ \beta }(u^{\beta })\overset{\alpha }{d}u^{\beta
^{\prime }},$ where the fractional local coordinate co--basis is
\begin{equation}
\ _{\ }\overset{\alpha }{d}u^{\beta ^{\prime }}=\left( (dx^{i^{\prime
}})^{\alpha },(dy^{a^{\prime }})^{\alpha }\right) .  \label{frlccb}
\end{equation}

It is possible to define a nonlinear connection (N--connection) $\overset{%
\alpha }{\mathbf{N}}$ \ for a fractional space $\overset{\alpha }{\mathbf{V}}
$ by a nonholonomic distribution (Whitney sum) with conventional h-- and
v--subspaces, $\underline{h}\overset{\alpha }{\mathbf{V}}$ and $\underline{v}%
\overset{\alpha }{\mathbf{V}},$%
\begin{equation}
\overset{\alpha }{\underline{T}}\overset{\alpha }{\mathbf{V}}=\underline{h}%
\overset{\alpha }{\mathbf{V}}\mathbf{\oplus }\underline{v}\overset{\alpha }{%
\mathbf{V}}.  \label{whit}
\end{equation}%
Locally, such a fractional N--connection is characterized by its local
coefficients $\overset{\alpha }{\mathbf{N}}\mathbf{=}\{\ ^{\alpha
}N_{i}^{a}\},$ when $ \overset{\alpha }{\mathbf{N}}\mathbf{=}\ ^{\alpha
}N_{i}^{a}(u)(dx^{i})^{\alpha }\otimes \overset{\alpha }{\underline{\partial
}}_{a}.$

On $\overset{\alpha }{\mathbf{V}},$ it is convenient to work with N--adapted
fractional (co) frames,
\begin{eqnarray}
\ ^{\alpha }\mathbf{e}_{\beta } &=&\left[ \ ^{\alpha }\mathbf{e}_{j}=\overset%
{\alpha }{\underline{\partial }}_{j}-\ ^{\alpha }N_{j}^{a}\overset{\alpha }{%
\underline{\partial }}_{a},\ ^{\alpha }e_{b}=\overset{\alpha }{\underline{%
\partial }}_{b}\right] ,  \label{dder} \\
\ ^{\alpha }\mathbf{e}^{\beta } &=&[\ ^{\alpha }e^{j}=(dx^{j})^{\alpha },\
^{\alpha }\mathbf{e}^{b}=(dy^{b})^{\alpha }+\ ^{\alpha
}N_{k}^{b}(dx^{k})^{\alpha }].  \label{ddif}
\end{eqnarray}

A fractional metric structure (d--metric) $\ \overset{\alpha }{\mathbf{g}}%
=\{\ ^{\alpha }g_{\underline{\alpha }\underline{\beta }}\}=\left[ \ ^{\alpha
}g_{kj},\ ^{\alpha }g_{cb}\right] $ on $\overset{\alpha }{\mathbf{V}}$ \ can
be represented in different equivalent forms,
\begin{eqnarray}
\ \overset{\alpha }{\mathbf{g}} &=&\ ^{\alpha }g_{\underline{\gamma }%
\underline{\beta }}(u)(du^{\underline{\gamma }})^{\alpha }\otimes (du^{%
\underline{\beta }})^{\alpha }  \label{m1} \\
&=&\ ^{\alpha }g_{kj}(x,y)\ ^{\alpha }e^{k}\otimes \ ^{\alpha }e^{j}+\
^{\alpha }g_{cb}(x,y)\ ^{\alpha }\mathbf{e}^{c}\otimes \ ^{\alpha }\mathbf{e}%
^{b}  \notag \\
&=&\eta _{k^{\prime }j^{\prime }}\ ^{\alpha }e^{k^{\prime }}\otimes \
^{\alpha }e^{j^{\prime }}+\eta _{c^{\prime }b^{\prime }}\ ^{\alpha }\mathbf{e%
}^{c^{\prime }}\otimes \ ^{\alpha }\mathbf{e}^{b^{\prime }},  \notag
\end{eqnarray}%
where matrices $\eta _{k^{\prime }j^{\prime }}=diag[\pm 1,\pm 1,...,\pm 1]$
and $\eta _{a^{\prime }b^{\prime }}=diag[\pm 1,\pm 1,...,\pm 1],$ for the
signature of a ''prime'' spacetime $\mathbf{V,}$ are obtained by frame
transforms $\eta _{k^{\prime }j^{\prime }}=e_{\ k^{\prime }}^{k}\ e_{\
j^{\prime }}^{j}\ _{\ }^{\alpha }g_{kj}$ and $\eta _{a^{\prime }b^{\prime
}}=e_{\ a^{\prime }}^{a}\ e_{\ b^{\prime }}^{b}\ _{\ }^{\alpha }g_{ab}.$

We can adapt geometric objects on $\overset{\alpha }{\mathbf{V}}$ with
respect to a given N--connection structure $\overset{\alpha }{\mathbf{N}},$
calling them as distinguished objects (d--objects). For instance, a
distinguished connection (d--connection) $\overset{\alpha }{\mathbf{D}}$ on $%
\overset{\alpha }{\mathbf{V}}$ is defined as a linear connection preserving
under parallel transports the Whitney sum (\ref{whit}). There is an
associated N--adapted differential 1--form
\begin{equation}
\ ^{\alpha }\mathbf{\Gamma }_{\ \beta }^{\tau }=\ ^{\alpha }\mathbf{\Gamma }%
_{\ \beta \gamma }^{\tau }\ ^{\alpha }\mathbf{e}^{\gamma },  \label{fdcf}
\end{equation}%
parametrizing the coefficients (with respect to (\ref{ddif}) and (\ref{dder}%
)) in the form $\ ^{\alpha }\mathbf{\Gamma }_{\ \tau \beta }^{\gamma
}=\left( \ ^{\alpha }L_{jk}^{i},\ ^{\alpha }L_{bk}^{a},\ ^{\alpha
}C_{jc}^{i},\ ^{\alpha }C_{bc}^{a}\right) .$

The absolute fractional differential $\ ^{\alpha }\mathbf{d}=\ _{\ _{1}x}%
\overset{\alpha }{d}_{x}+\ _{\ _{1}y}\overset{\alpha }{d}_{y}$ acts on
fractional differential forms in N--adapted form. This is a fractional
distinguished operator, d--operator, when the value $\ ^{\alpha }\mathbf{d:=}%
\ ^{\alpha }\mathbf{e}^{\beta }\ ^{\alpha }\mathbf{e}_{\beta }$ splits into
exterior h- and v--derivatives when
\begin{equation*}
\ _{\ _{1}x}\overset{\alpha }{d}_{x}:=(dx^{i})^{\alpha }\ \ _{\ _{1}x}%
\overset{\alpha }{\underline{\partial }}_{i}=\ ^{\alpha }e^{j}\ ^{\alpha }%
\mathbf{e}_{j}\mbox{ and }_{\ _{1}y}\overset{\alpha }{d}_{y}:=(dy^{a})^{%
\alpha }\ \ _{\ _{1}x}\overset{\alpha }{\underline{\partial }}_{a}=\
^{\alpha }\mathbf{e}^{b}\ ^{\alpha }e_{b}.
\end{equation*}%
Using such differentials, we can compute in explicit form the torsion and
curvature (as fractional two d--forms derived for (\ref{fdcf})) of a
fractional d--connection $\overset{\alpha }{\mathbf{D}}=\{\ ^{\alpha }%
\mathbf{\Gamma }_{\ \beta \gamma }^{\tau }\},$
\begin{eqnarray}
\ ^{\alpha }\mathcal{T}^{\tau } &\doteqdot &\overset{\alpha }{\mathbf{D}}\
^{\alpha }\mathbf{e}^{\tau }=\ ^{\alpha }\mathbf{d}\ ^{\alpha }\mathbf{e}%
^{\tau }+\ ^{\alpha }\mathbf{\Gamma }_{\ \beta }^{\tau }\wedge \ ^{\alpha }%
\mathbf{e}^{\beta }\mbox{ and }  \label{tors} \\
\ ^{\alpha }\mathcal{R}_{~\beta }^{\tau } &\doteqdot &\overset{\alpha }{%
\mathbf{D}}\mathbf{\ ^{\alpha }\Gamma }_{\ \beta }^{\tau }=\ ^{\alpha }%
\mathbf{d\ ^{\alpha }\Gamma }_{\ \beta }^{\tau }-\ ^{\alpha }\mathbf{\Gamma }%
_{\ \beta }^{\gamma }\wedge \ ^{\alpha }\mathbf{\Gamma }_{\ \gamma }^{\tau
}=\ ^{\alpha }\mathbf{R}_{\ \beta \gamma \delta }^{\tau }\ ^{\alpha }\mathbf{%
e}^{\gamma }\wedge \ ^{\alpha }\mathbf{e}^{\delta }.  \notag
\end{eqnarray}

Contracting respectively the indices, we can compute the fractional Ricci
tensor $\ ^{\alpha }\mathcal{R}ic=\{\ ^{\alpha }\mathbf{R}_{\alpha \beta
}\doteqdot \ ^{\alpha }\mathbf{R}_{\ \alpha \beta \tau }^{\tau }\}$ with
components
\begin{equation}
\ ^{\alpha }R_{ij}\doteqdot \ ^{\alpha }R_{\ ijk}^{k},\ \ \ ^{\alpha
}R_{ia}\doteqdot -\ ^{\alpha }R_{\ ika}^{k},\ \ ^{\alpha }R_{ai}\doteqdot \
^{\alpha }R_{\ aib}^{b},\ \ ^{\alpha }R_{ab}\doteqdot \ ^{\alpha }R_{\
abc}^{c}  \label{dricci}
\end{equation}%
and the scalar curvature of fractional d--connection $\overset{\alpha }{%
\mathbf{D}},$
\begin{equation}
\ _{s}^{\alpha }\mathbf{R}\doteqdot \ ^{\alpha }\mathbf{g}^{\tau \beta }\
^{\alpha }\mathbf{R}_{\tau \beta }=\ ^{\alpha }R+\ ^{\alpha }S,\ ^{\alpha
}R=\ ^{\alpha }g^{ij}\ ^{\alpha }R_{ij},\ \ ^{\alpha }S=\ ^{\alpha }g^{ab}\
^{\alpha }R_{ab},  \label{sdccurv}
\end{equation}%
with $\ ^{\alpha }\mathbf{g}^{\tau \beta }$ being the inverse coefficients
to a d--metric (\ref{m1}).\footnote{%
For applications in modern gravity and geometric mechanics, we can
considered more special classes of d--connections:
\par
\begin{itemize}
\item There is a unique canonical metric compatible fractional d--connection
$\ ^{\alpha }\widehat{\mathbf{D}}=\{\ ^{\alpha }\widehat{\mathbf{\Gamma }}%
_{\ \alpha \beta }^{\gamma }=\left( \ ^{\alpha }\widehat{L}_{jk}^{i},\
^{\alpha }\widehat{L}_{bk}^{a},\ ^{\alpha }\widehat{C}_{jc}^{i},\ ^{\alpha }%
\widehat{C}_{bc}^{a}\right) \},$ when $\ ^{\alpha }\widehat{\mathbf{D}}\
\left( \ ^{\alpha }\mathbf{g}\right) =0,$ satisfying the conditions $\
^{\alpha }\widehat{T}_{\ jk}^{i}=0$ and $\ ^{\alpha }\widehat{T}_{\
bc}^{a}=0,$ but $\ ^{\alpha }\widehat{T}_{\ ja}^{i},\ ^{\alpha }\widehat{T}%
_{\ ji}^{a}$ and $\ ^{\alpha }\widehat{T}_{\ bi}^{a}$ are not zero. The
N--adapted coefficients are explicitly determined by the coefficients of (%
\ref{m1}),
\begin{eqnarray*}
\ ^{\alpha }\widehat{L}_{jk}^{i} &=&\frac{1}{2}\ ^{\alpha }g^{ir}\left( \
^{\alpha }\mathbf{e}_{k}\ ^{\alpha }g_{jr}+\ ^{\alpha }\mathbf{e}_{j}\
^{\alpha }g_{kr}-\ ^{\alpha }\mathbf{e}_{r}\ ^{\alpha }g_{jk}\right) ,
\notag \\
\ ^{\alpha }\widehat{L}_{bk}^{a} &=&\ ^{\alpha }e_{b}(\ ^{\alpha }N_{k}^{a})+%
\frac{1}{2}\ ^{\alpha }g^{ac}(\ ^{\alpha }\mathbf{e}_{k}\ ^{\alpha }g_{bc}-\
^{\alpha }g_{dc}\ \ ^{\alpha }e_{b}\ ^{\alpha }N_{k}^{d}-\ ^{\alpha }g_{db}\
\ ^{\alpha }e_{c}\ ^{\alpha }N_{k}^{d}),  \notag \\
\ ^{\alpha }\widehat{C}_{jc}^{i} &=&\frac{1}{2}\ ^{\alpha }g^{ik}\ ^{\alpha
}e_{c}\ ^{\alpha }g_{jk},\ \ ^{\alpha }\widehat{C}_{bc}^{a} = \frac{1}{2}\
^{\alpha }g^{ad}\left( \ ^{\alpha }e_{c}\ ^{\alpha }g_{bd}+\ ^{\alpha
}e_{c}\ ^{\alpha }g_{cd}-\ ^{\alpha }e_{d}\ ^{\alpha }g_{bc}\right) .
\label{candcon}
\end{eqnarray*}%
\par
\item The fractional Levi--Civita connection $\ ^{\alpha }\nabla =\{\ \
^{\alpha }\Gamma _{\ \alpha \beta }^{\gamma }\}$ can be defined in standard
from but for the fractional Caputo left derivatives acting on the
coefficients of a fractional metric (\ref{m1}).
\end{itemize}
}

The Einstein tensor of any metric compatible $\overset{\alpha }{\mathbf{D}},
$ when $\overset{\alpha }{\mathbf{D}}_{\tau }\ ^{\alpha }\mathbf{g}^{\tau
\beta }=0,$ is defined $\ ^{\alpha }\mathcal{E}ns=\{\ ^{\alpha }\mathbf{G}%
_{\alpha \beta }\},$ where
\begin{equation}
\ ^{\alpha }\mathbf{G}_{\alpha \beta }:=\ ^{\alpha }\mathbf{R}_{\alpha \beta
}-\frac{1}{2}\ ^{\alpha }\mathbf{g}_{\alpha \beta }\ \ _{s}^{\alpha }\mathbf{%
R.}  \label{enstdt}
\end{equation}

The regular fractional mechanics defined by a fractional Lagrangian $\overset%
{\alpha }{L}$ can be equivalently encoded into canonical geometric data $(\
_{L}\overset{\alpha }{\mathbf{N}},\ _{L}\overset{\alpha }{\mathbf{g}},\
_{c}^{\alpha }\mathbf{D}),$ where we put the label $L$ in order to emphasize
that such geometric objects are induced by a fractional Lagrangian as we
provided in \cite{vrfrf,vrfrg,bv1,bv3}. We also note that it is possible to
''arrange'' on $\overset{\alpha }{\mathbf{V}}$ such nonholonomic
distributions when a d--connection $\ \ _{0}\overset{\alpha }{\mathbf{D}}%
=\{\ _{0}^{\alpha }\widetilde{\mathbf{\Gamma }}_{\ \alpha ^{\prime }\beta
^{\prime }}^{\gamma ^{\prime }}\}$ is described by constant matrix
coefficients, see details in \cite{vacap,vanco}, for integer dimensions, and %
\cite{bv3}, for fractional dimensions.

\section{Basic Equations for Fractional Curve Flows}

\label{s3} In symbolic, abstract index form, the constructions for
nonholonomic fractional spaces with correspondingly defined distributions
are similar to those for the Riemannian symmetric--spaces soldered to Klein
geometry of ''integer'' dimension. The fractional structure is encoded into
the local Caputo derivatives.

Following the introduced Cartan--Killing parametrizations, we analyze the
flow $\gamma (\tau ,\mathbf{l})$ of a non--stretching curve in $\mathbf{V}_{%
\mathbf{N}}=\mathbf{G}/SO(n)\oplus $ $SO(m)$ extended to $\overset{\alpha }{%
\mathbf{V}}_{\mathbf{N}}=\mathbf{G}/SO(n)\oplus $ $SO(m).$\footnote{%
We use an isomorphism between the real space $\mathfrak{so}(n)$ and the Lie
algebra of $n\times n$ skew--symmetric matrices. This allows us to establish
an isomorphism between $h\mathfrak{p}$ $\simeq \mathbb{R}^{n}$ and the
tangent spaces $T_{x}M=\mathfrak{so}(n+1)/$ $\mathfrak{so}(n)$ of the
Riemannian manifold $M=SO(n+1)/$ $SO(n)$ as described by the following
canonical decomposition
\begin{equation*}
h\mathfrak{g}=\mathfrak{so}(n+1)\supset h\mathfrak{p\in }\left[
\begin{array}{cc}
0 & h\mathbf{p} \\
-h\mathbf{p}^{T} & h\mathbf{0}%
\end{array}%
\right] \mbox{\ for\ }h\mathbf{0\in }h\mathfrak{h=so}(n)
\end{equation*}%
with $h\mathbf{p=\{}p^{i^{\prime }}\mathbf{\}\in }\mathbb{R}^{n}$ being the
h--component of the d--vector $\mathbf{p=(}p^{i^{\prime }}\mathbf{,}%
p^{a^{\prime }}\mathbf{)}$ and $h\mathbf{p}^{T}$ mean the transposition of
the row $h\mathbf{p.}$ \ In our approach, $T_{x}M\rightarrow \underline{T}%
_{x}M,$ with Caputo fractional derivatives. The Cartan--Killing inner
product on $h\mathfrak{g}$ is
\begin{eqnarray*}
h\mathbf{p\cdot }h\mathbf{p} &\mathbf{=}&\left\langle \left[
\begin{array}{cc}
0 & h\mathbf{p} \\
-h\mathbf{p}^{T} & h\mathbf{0}%
\end{array}%
\right] ,\left[
\begin{array}{cc}
0 & h\mathbf{p} \\
-h\mathbf{p}^{T} & h\mathbf{0}%
\end{array}%
\right] \right\rangle \\
&\mathbf{\doteqdot }&\frac{1}{2}tr\left\{ \left[
\begin{array}{cc}
0 & h\mathbf{p} \\
-h\mathbf{p}^{T} & h\mathbf{0}%
\end{array}%
\right] ^{T}\left[
\begin{array}{cc}
0 & h\mathbf{p} \\
-h\mathbf{p}^{T} & h\mathbf{0}%
\end{array}%
\right] \right\} ,
\end{eqnarray*}%
where $tr$ denotes the trace of product of matrices. This product identifies
canonically $h\mathfrak{p}$ $\simeq \mathbb{R}^{n}$ with its dual $h%
\mathfrak{p}^{\ast }$ $\simeq \mathbb{R}^{n}.$ In a similar form, we can
consider
\begin{equation*}
v\mathfrak{g}=\mathfrak{so}(m+1)\supset v\mathfrak{p\in }\left[
\begin{array}{cc}
0 & v\mathbf{p} \\
-v\mathbf{p}^{T} & v\mathbf{0}%
\end{array}%
\right] \mbox{\ for\ }v\mathbf{0\in }v\mathfrak{h=so}(m)
\end{equation*}%
with $v\mathbf{p=\{}p^{a^{\prime }}\mathbf{\}\in }\mathbb{R}^{m}$ being the
v--component of the d--vector $\mathbf{p=(}p^{i^{\prime }}\mathbf{,}%
p^{a^{\prime }}\mathbf{)}$ and define the Cartan--Killing inner product $v%
\mathbf{p\cdot }v\mathbf{p\doteqdot }\frac{1}{2}tr\{...\}.$ In general, we
can consider the Cartan--Killing N--adapted inner product $\mathbf{p\cdot p=}%
h\mathbf{p\cdot }h\mathbf{p+}v\mathbf{p\cdot }v\mathbf{p.}$} This extension
is defined via a coframe $\ ^{\alpha }\mathbf{e}\in \underline{T}_{\gamma
}^{\ast }\mathbf{V}_{\mathbf{N}}\otimes (h\mathfrak{p\oplus }v\mathfrak{p}),$
which is a N--adapted $\left( SO(n)\mathfrak{\oplus }SO(m)\right) $%
--parallel basis along $\gamma ,$ and its associated canonical
d--con\-nec\-tion 1--form $\ ^{\alpha }\mathbf{\Gamma }\in \underline{T}%
_{\gamma }^{\ast }\mathbf{V}_{\mathbf{N}}\otimes (\mathfrak{so}(n)\mathfrak{%
\oplus so}(m)).$ Such d--objects are parametrized:%
\begin{equation*}
\ ^{\alpha }\mathbf{e}_{\mathbf{X}}=\ ^{\alpha }\mathbf{e}_{h\mathbf{X}}+\
^{\alpha }\mathbf{e}_{v\mathbf{X}},
\end{equation*}%
where (for $(1,\overrightarrow{0})\in \mathbb{R}^{n},\overrightarrow{0}\in
\mathbb{R}^{n-1}$ and $(1,\overleftarrow{0})\in \mathbb{R}^{m},%
\overleftarrow{0}\in \mathbb{R}^{m-1}),$
\begin{equation*}
\ ^{\alpha }\mathbf{e}_{h\mathbf{X}}=\gamma _{h\mathbf{X}}\rfloor h\
^{\alpha }\mathbf{e=}\left[
\begin{array}{cc}
0 & (1,\overrightarrow{0}) \\
-(1,\overrightarrow{0})^{T} & h\mathbf{0}%
\end{array}%
\right]
\end{equation*}%
and
\begin{equation*}
\ ^{\alpha }\mathbf{e}_{v\mathbf{X}}=\gamma _{v\mathbf{X}}\rfloor v\
^{\alpha }\mathbf{e=}\left[
\begin{array}{cc}
0 & (1,\overleftarrow{0}) \\
-(1,\overleftarrow{0})^{T} & v\mathbf{0}%
\end{array}%
\right] .
\end{equation*}
We also introduce%
\begin{equation*}
\ ^{\alpha }\mathbf{\Gamma =}\left[ \mathbf{\Gamma }_{h\mathbf{X}},\mathbf{%
\Gamma }_{v\mathbf{X}}\right] ,
\end{equation*}%
for
\begin{equation*}
\mathbf{\Gamma }_{h\mathbf{X}}\mathbf{=}\gamma _{h\mathbf{X}}\rfloor \mathbf{%
L=}\left[
\begin{array}{cc}
0 & (0,\overrightarrow{0}) \\
-(0,\overrightarrow{0})^{T} & \mathbf{L}%
\end{array}%
\right] \in \mathfrak{so}(n+1),
\end{equation*}%
where $\mathbf{L=}\left[
\begin{array}{cc}
0 & \overrightarrow{v} \\
-\overrightarrow{v}^{T} & h\mathbf{0}%
\end{array}%
\right] \in \mathfrak{so}(n),~\overrightarrow{v}\in \mathbb{R}^{n-1},~h%
\mathbf{0\in }\mathfrak{so}(n-1),$ and
\begin{equation*}
\mathbf{\Gamma }_{v\mathbf{X}}\mathbf{=}\gamma _{v\mathbf{X}}\rfloor \mathbf{%
C=}\left[
\begin{array}{cc}
0 & (0,\overleftarrow{0}) \\
-(0,\overleftarrow{0})^{T} & \mathbf{C}%
\end{array}%
\right] \in \mathfrak{so}(m+1),
\end{equation*}%
where $\mathbf{C=}\left[
\begin{array}{cc}
0 & \overleftarrow{v} \\
-\overleftarrow{v}^{T} & v\mathbf{0}%
\end{array}%
\right] \in \mathfrak{so}(m),~\overleftarrow{v}\in \mathbb{R}^{m-1},~v%
\mathbf{0\in }\mathfrak{so}(m-1).$

There are decompositions of horizontal $SO(n+1)/$ $SO(n)$ matrices,
\begin{eqnarray*}
h\mathfrak{p} &\mathfrak{\ni }&\left[
\begin{array}{cc}
0 & h\mathbf{p} \\
-h\mathbf{p}^{T} & h\mathbf{0}%
\end{array}%
\right] =\left[
\begin{array}{cc}
0 & \left( h\mathbf{p}_{\parallel },\overrightarrow{0}\right) \\
-\left( h\mathbf{p}_{\parallel },\overrightarrow{0}\right) ^{T} & h\mathbf{0}%
\end{array}%
\right] \\
&&+\left[
\begin{array}{cc}
0 & \left( 0,h\overrightarrow{\mathbf{p}}_{\perp }\right) \\
-\left( 0,h\overrightarrow{\mathbf{p}}_{\perp }\right) ^{T} & h\mathbf{0}%
\end{array}%
\right] ,
\end{eqnarray*}%
into tangential and normal parts relative to $\ ^{\alpha }\mathbf{e}_{h%
\mathbf{X}}$ via corresponding decompositions of h--vectors $h\mathbf{p=(}h%
\mathbf{\mathbf{p}_{\parallel },}h\mathbf{\overrightarrow{\mathbf{p}}_{\perp
})\in }\mathbb{R}^{n}$ relative to $\left( 1,\overrightarrow{0}\right) ,$
when $h\mathbf{\mathbf{p}_{\parallel }}$ is identified with $h\mathfrak{p}%
_{C}$ and $h\mathbf{\overrightarrow{\mathbf{p}}_{\perp }}$ is identified
with $h\mathfrak{p}_{\perp }=h\mathfrak{p}_{C^{\perp }}.$ In a similar form,
it is possible to decompose vertical $SO(m+1)/$ $SO(m)$ matrices,
\begin{eqnarray*}
v\mathfrak{p} &\mathfrak{\ni }&\left[
\begin{array}{cc}
0 & v\mathbf{p} \\
-v\mathbf{p}^{T} & v\mathbf{0}%
\end{array}%
\right] =\left[
\begin{array}{cc}
0 & \left( v\mathbf{p}_{\parallel },\overleftarrow{0}\right) \\
-\left( v\mathbf{p}_{\parallel },\overleftarrow{0}\right) ^{T} & v\mathbf{0}%
\end{array}%
\right] \\
&&+\left[
\begin{array}{cc}
0 & \left( 0,v\overleftarrow{\mathbf{p}}_{\perp }\right) \\
-\left( 0,v\overleftarrow{\mathbf{p}}_{\perp }\right) ^{T} & v\mathbf{0}%
\end{array}%
\right] ,
\end{eqnarray*}%
into tangential and normal parts relative to $\ ^{\alpha }\mathbf{e}_{v%
\mathbf{X}}$ via corresponding decompositions of h--vectors $v\mathbf{p=(}v%
\mathbf{\mathbf{p}_{\parallel },}v\overleftarrow{\mathbf{\mathbf{p}}}\mathbf{%
_{\perp })\in }\mathbb{R}^{m}$ relative to $\left( 1,\overleftarrow{0}%
\right) ,$ when $v\mathbf{\mathbf{p}_{\parallel }}$ is identified with $v%
\mathfrak{p}_{C}$ and $v\overleftarrow{\mathbf{\mathbf{p}}}\mathbf{_{\perp }}
$ is identified with $v\mathfrak{p}_{\perp }=v\mathfrak{p}_{C^{\perp }}.$
Locally, we consider, for instance, instead of $\mathbb{R}^{n}$ the
fractional space $\ ^{\alpha }\mathbb{R}^{n}$ local (co) vectors defined by
Caputo fractional derivatives and their duals.

The canonical d--connection $\ ^{\alpha }\widehat{\mathbf{D}}=\{\ ^{\alpha }%
\widehat{\mathbf{\Gamma }}_{\ \beta \gamma }^{\tau }\}$ \cite{bv3} induces
matrices decomposed with respect to the fractional flow direction. In the
h--direction, we parametrize%
\begin{equation*}
\ ^{\alpha }\mathbf{e}_{h\mathbf{Y}}=\gamma _{\tau }\rfloor h\ ^{\alpha }%
\mathbf{e=}\left[
\begin{array}{cc}
0 & \left( h\mathbf{e}_{\parallel },h\overrightarrow{\mathbf{e}}_{\perp
}\right) \\
-\left( h\mathbf{e}_{\parallel },h\overrightarrow{\mathbf{e}}_{\perp
}\right) ^{T} & h\mathbf{0}%
\end{array}%
\right] ,
\end{equation*}%
when $\ ^{\alpha }\mathbf{e}_{h\mathbf{Y}}\in h\mathfrak{p,}\left( h\mathbf{e%
}_{\parallel },h\overrightarrow{\mathbf{e}}_{\perp }\right) \in \ ^{\alpha }%
\mathbb{R}^{n}$ and $h\overrightarrow{\mathbf{e}}_{\perp }\in \ ^{\alpha }%
\mathbb{R}^{n-1},$ and
\begin{equation}
\mathbf{\Gamma }_{h\mathbf{Y}}\mathbf{=}\gamma _{h\mathbf{Y}}\rfloor \mathbf{%
L=}\left[
\begin{array}{cc}
0 & (0,\overrightarrow{0}) \\
-(0,\overrightarrow{0})^{T} & h\mathbf{\varpi }_{\tau }%
\end{array}%
\right] \in \mathfrak{so}(n+1),  \label{auxaaa}
\end{equation}%
where $\ \ h\mathbf{\varpi }_{\tau }\mathbf{=}\left[
\begin{array}{cc}
0 & \overrightarrow{\varpi } \\
-\overrightarrow{\varpi }^{T} & h\mathbf{\Theta }%
\end{array}%
\right] \in \mathfrak{so}(n),~\overrightarrow{\varpi }\in \ ^{\alpha }%
\mathbb{R}^{n-1},~h\mathbf{\Theta \in }\mathfrak{so}(n-1).$

In the v--direction, we have
\begin{equation*}
\mathbf{e}_{v\mathbf{Y}}=\gamma _{\tau }\rfloor v\mathbf{e=}\left[
\begin{array}{cc}
0 & \left( v\mathbf{e}_{\parallel },v\overleftarrow{\mathbf{e}}_{\perp
}\right) \\
-\left( v\mathbf{e}_{\parallel },v\overleftarrow{\mathbf{e}}_{\perp }\right)
^{T} & v\mathbf{0}%
\end{array}%
\right] ,
\end{equation*}%
when $\mathbf{e}_{v\mathbf{Y}}\in v\mathfrak{p,}\left( v\mathbf{e}%
_{\parallel },v\overleftarrow{\mathbf{e}}_{\perp }\right) \in \ ^{\alpha }%
\mathbb{R}^{m}$ and $v\overleftarrow{\mathbf{e}}_{\perp }\in \ ^{\alpha }%
\mathbb{R}^{m-1},$ and
\begin{equation*}
\mathbf{\Gamma }_{v\mathbf{Y}}\mathbf{=}\gamma _{v\mathbf{Y}}\rfloor \mathbf{%
C=}\left[
\begin{array}{cc}
0 & (0,\overleftarrow{0}) \\
-(0,\overleftarrow{0})^{T} & v\mathbf{\varpi }_{\tau }%
\end{array}%
\right] \in \mathfrak{so}(m+1),
\end{equation*}%
where $v\mathbf{\varpi }_{\tau }\mathbf{=}\left[
\begin{array}{cc}
0 & \overleftarrow{\varpi } \\
-\overleftarrow{\varpi }^{T} & v\mathbf{\Theta }%
\end{array}%
\right] \in \mathfrak{so}(m),~\overleftarrow{\varpi }\in \ ^{\alpha }\mathbb{%
R}^{m-1},~v\mathbf{\Theta \in }\mathfrak{so}(m-1).$

The components $h\mathbf{e}_{\parallel }$ and $h\overrightarrow{\mathbf{e}}%
_{\perp }$ correspond to the decomposition
\begin{equation*}
\ ^{\alpha }\mathbf{e}_{h\mathbf{Y}}=h\mathbf{g(\gamma }_{\tau },\mathbf{%
\gamma }_{\mathbf{l}}\mathbf{)}\ ^{\alpha }\mathbf{e}_{h\mathbf{X}}+\mathbf{%
(\gamma }_{\tau })_{\perp }\rfloor h\mathbf{e}_{\perp }
\end{equation*}%
into tangential and normal parts relative to $\ ^{\alpha }\mathbf{e}_{h%
\mathbf{X}}.$ In a similar form, one considers $v\mathbf{e}_{\parallel }$
and $v\overleftarrow{\mathbf{e}}_{\perp }$ corresponding to the
decomposition
\begin{equation*}
\ ^{\alpha }\mathbf{e}_{v\mathbf{Y}}=v\mathbf{g(\gamma }_{\tau },\mathbf{%
\gamma }_{\mathbf{l}}\mathbf{)}\ ^{\alpha }\mathbf{e}_{v\mathbf{X}}+\mathbf{%
(\gamma }_{\tau })_{\perp }\rfloor v\mathbf{e}_{\perp }.
\end{equation*}%
Working with such matrix parametrizations, we define%
\begin{eqnarray}
\left[ \ ^{\alpha }\mathbf{e}_{h\mathbf{X}},\ ^{\alpha }\mathbf{e}_{h\mathbf{%
Y}}\right] &=&-\left[
\begin{array}{cc}
0 & 0 \\
0 & h\mathbf{e}_{\perp }%
\end{array}%
\right] \in \mathfrak{so}(n+1),  \label{aux41} \\
\mbox{ \ for \ }h\mathbf{e}_{\perp } &=&\left[
\begin{array}{cc}
0 & h\overrightarrow{\mathbf{e}}_{\perp } \\
-(h\overrightarrow{\mathbf{e}}_{\perp })^{T} & h\mathbf{0}%
\end{array}%
\right] \in \mathfrak{so}(n);  \notag \\
\left[ \mathbf{\Gamma }_{h\mathbf{Y}},\ ^{\alpha }\mathbf{e}_{h\mathbf{Y}}%
\right] &=&-\left[
\begin{array}{cc}
0 & \left( 0,\overrightarrow{\varpi }\right) \\
-\left( 0,\overrightarrow{\varpi }\right) ^{T} & 0%
\end{array}%
\right] \in h\mathfrak{p}_{\perp };  \notag \\
\left[ \mathbf{\Gamma }_{h\mathbf{X}},\ ^{\alpha }\mathbf{e}_{h\mathbf{Y}}%
\right] &=&-\left[
\begin{array}{cc}
0 & \left( -\overrightarrow{v}\cdot h\overrightarrow{\mathbf{e}}_{\perp },h%
\mathbf{e}_{\parallel }\overrightarrow{v}\right) \\
-\left( -\overrightarrow{v}\cdot h\overrightarrow{\mathbf{e}}_{\perp },h%
\mathbf{e}_{\parallel }\overrightarrow{v}\right) ^{T} & h\mathbf{0}%
\end{array}%
\right]  \notag \\
& & \in h\mathfrak{p};  \notag
\end{eqnarray}%
and
\begin{eqnarray}
\left[ \ ^{\alpha }\mathbf{e}_{v\mathbf{X}},\ ^{\alpha }\mathbf{e}_{v\mathbf{%
Y}}\right] &=&-\left[
\begin{array}{cc}
0 & 0 \\
0 & v\mathbf{e}_{\perp }%
\end{array}%
\right] \in \mathfrak{so}(m+1),  \label{aux41a} \\
\mbox{ \ for \ }v\mathbf{e}_{\perp } &=&\left[
\begin{array}{cc}
0 & v\overrightarrow{\mathbf{e}}_{\perp } \\
-(v\overrightarrow{\mathbf{e}}_{\perp })^{T} & v\mathbf{0}%
\end{array}%
\right] \in \mathfrak{so}(m);  \notag \\
\left[ \mathbf{\Gamma }_{v\mathbf{Y}},\ ^{\alpha }\mathbf{e}_{v\mathbf{Y}}%
\right] &=&-\left[
\begin{array}{cc}
0 & \left( 0,\overleftarrow{\varpi }\right) \\
-\left( 0,\overleftarrow{\varpi }\right) ^{T} & 0%
\end{array}%
\right] \in v\mathfrak{p}_{\perp };  \notag \\
\left[ \mathbf{\Gamma }_{v\mathbf{X}},\ ^{\alpha }\mathbf{e}_{v\mathbf{Y}}%
\right] &=&-\left[
\begin{array}{cc}
0 & \left( -\overleftarrow{v}\cdot v\overleftarrow{\mathbf{e}}_{\perp },v%
\mathbf{e}_{\parallel }\overleftarrow{v}\right) \\
-\left( -\overleftarrow{v}\cdot v\overleftarrow{\mathbf{e}}_{\perp },v%
\mathbf{e}_{\parallel }\overleftarrow{v}\right) ^{T} & v\mathbf{0}%
\end{array}%
\right]  \notag \\
&&\in v\mathfrak{p}.  \notag
\end{eqnarray}

We use formulas (\ref{aux41}) and (\ref{aux41a}) in order to write the
structure equations in terms of N--adapted curve fractional flow operators
soldered to the geometry Klein N--anholonomic spaces, the formulas are
''fractional'' extensions of the respective ones in Refs. \cite{vacap,vanco}%
. This way, it is possible to construct respectively the $\mathbf{G}$%
--invariant N--adapted torsion and curvature generated by the fractional
canonical d--connection,
\begin{eqnarray}
\ ^{\alpha }\widehat{\mathbf{T}}(\gamma _{\tau },\gamma _{\mathbf{l}})
&=&\left( \ ^{\alpha }\widehat{\mathbf{D}}_{\mathbf{X}}\gamma _{\tau }-\
^{\alpha }\widehat{\mathbf{D}}_{\mathbf{Y}}\gamma _{\mathbf{l}}\right)
\rfloor \ ^{\alpha }\mathbf{e}  \label{torscf} \\
&\mathbf{=}&\ ^{\alpha }\widehat{\mathbf{D}}_{\mathbf{X}}\ ^{\alpha }\mathbf{%
e}_{\mathbf{Y}}-\ ^{\alpha }\widehat{\mathbf{D}}_{\mathbf{Y}}\ ^{\alpha }%
\mathbf{e}_{\mathbf{X}}+\left[ \ ^{\alpha }\widehat{\mathbf{\Gamma }}_{%
\mathbf{X}},\ ^{\alpha }\mathbf{e}_{\mathbf{Y}}\right] -\left[ \ ^{\alpha }%
\widehat{\mathbf{\Gamma }}_{\mathbf{Y}},\ ^{\alpha }\mathbf{e}_{\mathbf{X}}%
\right]  \notag
\end{eqnarray}%
and
\begin{eqnarray}
\ ^{\alpha }\widehat{\mathbf{R}}(\gamma _{\tau },\gamma _{\mathbf{l}})\
^{\alpha }\mathbf{e} &\mathbf{=}&\left[ \ ^{\alpha }\widehat{\mathbf{D}}_{%
\mathbf{X}},\ ^{\alpha }\widehat{\mathbf{D}}_{\mathbf{Y}}\right] \ ^{\alpha }%
\mathbf{e}  \label{curvcf} \\
&\mathbf{=}&\ ^{\alpha }\widehat{\mathbf{D}}_{\mathbf{X}}\ ^{\alpha }%
\widehat{\mathbf{\Gamma }}_{\mathbf{Y}}-\ ^{\alpha }\widehat{\mathbf{D}}_{%
\mathbf{Y}}\ ^{\alpha }\widehat{\mathbf{\Gamma }}_{\mathbf{X}}+\left[ \
^{\alpha }\widehat{\mathbf{\Gamma }}_{\mathbf{X}},\ ^{\alpha }\widehat{%
\mathbf{\Gamma }}_{\mathbf{Y}}\right]  \notag
\end{eqnarray}%
where $\ ^{\alpha }\mathbf{e}_{\mathbf{X}}\doteqdot \gamma _{\mathbf{l}%
}\rfloor \ ^{\alpha }\mathbf{e,}$ $\ ^{\alpha }\mathbf{e}_{\mathbf{Y}%
}\doteqdot \gamma _{\mathbf{\tau }}\rfloor \ ^{\alpha }\mathbf{e,}$ $\
^{\alpha }\widehat{\mathbf{\Gamma }}_{\mathbf{X}}\doteqdot \gamma _{\mathbf{l%
}}\rfloor \ ^{\alpha }\widehat{\mathbf{\Gamma }}$ and $\ ^{\alpha }\widehat{%
\mathbf{\Gamma }}_{\mathbf{Y}}\doteqdot \gamma _{\mathbf{\tau }}\rfloor \
^{\alpha }\widehat{\mathbf{\Gamma }}\mathbf{.}$

Applying a d--connection $\ ^{\alpha }\mathbf{D}$ (in particular, we can
take $\ ^{\alpha }\widehat{\mathbf{D}}$) instead of the Levi--Civita one $\
^{\alpha }\nabla ,$ we get
\begin{eqnarray}
0 &=&\left( \ ^{\alpha }\mathbf{D}_{h\mathbf{X}}\gamma _{\tau }-\ ^{\alpha }%
\mathbf{D}_{h\mathbf{Y}}\gamma _{\mathbf{l}}\right) \rfloor h\ ^{\alpha }%
\mathbf{e}  \label{torseq} \\
&\mathbf{=}&\ ^{\alpha }\mathbf{D}_{h\mathbf{X}}\ ^{\alpha }\mathbf{e}_{h%
\mathbf{Y}}-\ ^{\alpha }\mathbf{D}_{h\mathbf{Y}}\ ^{\alpha }\mathbf{e}_{h%
\mathbf{X}}+\left[ \ ^{\alpha }\mathbf{L}_{h\mathbf{X}},\ ^{\alpha }\mathbf{e%
}_{h\mathbf{Y}}\right] -\left[ \ ^{\alpha }\mathbf{L}_{h\mathbf{Y}},\
^{\alpha }\mathbf{e}_{h\mathbf{X}}\right] ;  \notag \\
0 &=&\left( \ ^{\alpha }\mathbf{D}_{v\mathbf{X}}\gamma _{\tau }-\ ^{\alpha }%
\mathbf{D}_{v\mathbf{Y}}\gamma _{\mathbf{l}}\right) \rfloor v\ ^{\alpha }%
\mathbf{e}  \notag \\
&\mathbf{=}&\ ^{\alpha }\mathbf{D}_{v\mathbf{X}}\ ^{\alpha }\mathbf{e}_{v%
\mathbf{Y}}-\ ^{\alpha }\mathbf{D}_{v\mathbf{Y}}\ ^{\alpha }\mathbf{e}_{v%
\mathbf{X}}+\left[ \ ^{\alpha }\mathbf{C}_{v\mathbf{X}},\ ^{\alpha }\mathbf{e%
}_{v\mathbf{Y}}\right] -\left[ \ ^{\alpha }\mathbf{C}_{v\mathbf{Y}},\
^{\alpha }\mathbf{e}_{v\mathbf{X}}\right] ,  \notag \\
&&h\ ^{\alpha }\mathbf{R}(\gamma _{\tau },\gamma _{\mathbf{l}})h\ ^{\alpha }%
\mathbf{e} =\left[ \ ^{\alpha }\mathbf{D}_{h\mathbf{X}},\ ^{\alpha }\mathbf{D%
}_{h\mathbf{Y}}\right] h\ ^{\alpha }\mathbf{e}  \notag \\
&\mathbf{=}&\ ^{\alpha }\mathbf{D}_{h\mathbf{X}}\ ^{\alpha }\mathbf{L}_{h%
\mathbf{Y}}-\ ^{\alpha }\mathbf{D}_{h\mathbf{Y}}\ ^{\alpha }\mathbf{L}_{h%
\mathbf{X}}+\left[ \ ^{\alpha }\mathbf{L}_{h\mathbf{X}},\ ^{\alpha }\mathbf{L%
}_{h\mathbf{Y}}\right]  \notag \\
&&v\ ^{\alpha }\mathbf{R}(\gamma _{\tau },\gamma _{\mathbf{l}})v\ ^{\alpha }%
\mathbf{e} =\left[ \ ^{\alpha }\mathbf{D}_{v\mathbf{X}},\ ^{\alpha }\mathbf{D%
}_{v\mathbf{Y}}\right] v\ ^{\alpha }\mathbf{e}  \notag \\
&\mathbf{=}&\ ^{\alpha }\mathbf{D}_{v\mathbf{X}}\ ^{\alpha }\mathbf{C}_{v%
\mathbf{Y}}-\ ^{\alpha }\mathbf{D}_{v\mathbf{Y}}\ ^{\alpha }\mathbf{C}_{v%
\mathbf{X}}+\left[ \ ^{\alpha }\mathbf{C}_{v\mathbf{X}},\ ^{\alpha }\mathbf{C%
}_{v\mathbf{Y}}\right] .  \notag
\end{eqnarray}

Following N--adapted curve flow parametrizations (\ref{aux41}) and (\ref%
{aux41a}), the equations (\ref{torseq}) are written
\begin{eqnarray}
&&0=\ ^{\alpha }\mathbf{D}_{h\mathbf{X}}h\mathbf{e}_{\parallel }+%
\overrightarrow{v}\cdot h\overrightarrow{\mathbf{e}}_{\perp },~0=\ ^{\alpha }%
\mathbf{D}_{v\mathbf{X}}v\mathbf{e}_{\parallel }+\overleftarrow{v}\cdot v%
\overleftarrow{\mathbf{e}}_{\perp };  \label{torseqd} \\
&&0=\overrightarrow{\varpi }-h\mathbf{e}_{\parallel }\overrightarrow{v}+\
^{\alpha }\mathbf{D}_{h\mathbf{X}}h\overrightarrow{\mathbf{e}}_{\perp },~0=%
\overleftarrow{\varpi }-v\mathbf{e}_{\parallel }\overleftarrow{v}+\ ^{\alpha
}\mathbf{D}_{v\mathbf{X}}v\overleftarrow{\mathbf{e}}_{\perp };  \notag \\
&&\ ^{\alpha }\mathbf{D}_{h\mathbf{X}}\overrightarrow{\varpi }-\ ^{\alpha }%
\mathbf{D}_{h\mathbf{Y}}\overrightarrow{v}+\overrightarrow{v}\rfloor h%
\mathbf{\Theta }=h\overrightarrow{\mathbf{e}}_{\perp },  \notag \\
&&\ ^{\alpha }\mathbf{D}_{v\mathbf{X}}\overleftarrow{\varpi }-\ ^{\alpha }%
\mathbf{D}_{v\mathbf{Y}}\overleftarrow{v}+\overleftarrow{v}\rfloor v\mathbf{%
\Theta =}v\overleftarrow{\mathbf{e}}_{\perp };  \notag \\
&&\ ^{\alpha }\mathbf{D}_{h\mathbf{X}}h\mathbf{\Theta -}\overrightarrow{v}%
\otimes \overrightarrow{\varpi }+\overrightarrow{\varpi }\otimes
\overrightarrow{v}=0,  \notag \\
&& \ ^{\alpha }\mathbf{D}_{v\mathbf{X}}v\mathbf{\Theta -}\overleftarrow{v}%
\otimes \overleftarrow{\varpi }+\overleftarrow{\varpi }\otimes
\overleftarrow{v}=0.  \notag
\end{eqnarray}%
For such fractional spaces, the tensor and interior products, for instance,
for the h--components, are defined in the form: $\otimes $ denotes the outer
product of pairs of vectors ($1\times n$ row matrices), producing $n\times n$
matrices $\overrightarrow{A}\otimes \overrightarrow{B}=\overrightarrow{A}^{T}%
\overrightarrow{B},$ and $\rfloor $ denotes multiplication of $n\times n$
matrices on vectors ($1\times n$ row matrices); one holds the properties $%
\overrightarrow{A}\rfloor \left( \overrightarrow{B}\otimes \overrightarrow{C}%
\right) =\left( \overrightarrow{A}\cdot \overrightarrow{B}\right)
\overrightarrow{C}$ which is the transpose of the standard matrix product on
column vectors, and $\left( \overrightarrow{B}\otimes \overrightarrow{C}%
\right) \overrightarrow{A}=\left( \overrightarrow{C}\cdot \overrightarrow{A}%
\right) \overrightarrow{B}.$ As basic vectors, we use the Caputo fractional
derivatives. Similar formulas hold for the v--components but, for instance,
we have to change, correspondingly, $n\rightarrow m$ and $\overrightarrow{A}%
\rightarrow \overleftarrow{A};$ for such constructions the fractional
differentials have to be used.

\begin{lemma}
On nonholonomic fractional manifolds with constant curvature matrix
coefficients for a d--connection, there are N--adapted fractional
Hamiltonian symplectic operators,
\begin{equation}
h\ ^{\alpha }\mathcal{J}=\ ^{\alpha }\mathbf{D}_{h\mathbf{X}}+\ ^{\alpha }%
\mathbf{D}_{h\mathbf{X}}^{-1}\left( \overrightarrow{v}\cdot \right)
\overrightarrow{v}\mbox{ \ and \ }v\ ^{\alpha }\mathcal{J}=\ ^{\alpha }%
\mathbf{D}_{v\mathbf{X}}+\ ^{\alpha }\mathbf{D}_{v\mathbf{X}}^{-1}\left(
\overleftarrow{v}\cdot \right) \overleftarrow{v},  \label{sop}
\end{equation}%
and cosymplectic operators%
\begin{equation}
h\ ^{\alpha }\mathcal{H}\doteqdot \ ^{\alpha }\mathbf{D}_{h\mathbf{X}}+%
\overrightarrow{v}\rfloor \ ^{\alpha }\mathbf{D}_{h\mathbf{X}}^{-1}\left(
\overrightarrow{v}\wedge \right) \mbox{
\ and \ }v\ ^{\alpha }\mathcal{H}\doteqdot \ ^{\alpha }\mathbf{D}_{v\mathbf{X%
}}+\overleftarrow{v}\rfloor \ ^{\alpha }\mathbf{D}_{v\mathbf{X}}^{-1}\left(
\overleftarrow{v}\wedge \right) ,  \label{csop}
\end{equation}%
where, for instance, $\overrightarrow{A}\wedge \overrightarrow{B}=%
\overrightarrow{A}\otimes \overrightarrow{B}-\overrightarrow{B}\otimes $ $%
\overrightarrow{A}.$\
\end{lemma}

\begin{proof}
We sketch some key steps of the proof. The variables $\mathbf{e}_{\parallel
} $ and $\mathbf{\Theta ,}$ written in h-- and v--components, can be
expressed correspondingly in terms of variables $\overrightarrow{v},%
\overrightarrow{\varpi },h\overrightarrow{\mathbf{e}}_{\perp }$ and $%
\overleftarrow{v},\overleftarrow{\varpi },v\overleftarrow{\mathbf{e}}_{\perp
}$ (see equations (\ref{torseqd})),%
\begin{equation*}
h\mathbf{e}_{\parallel }=-\ ^{\alpha }\mathbf{D}_{h\mathbf{X}}^{-1}(%
\overrightarrow{v}\cdot h\overrightarrow{\mathbf{e}}_{\perp }),~v\mathbf{e}%
_{\parallel }=-\ ^{\alpha }\mathbf{D}_{v\mathbf{X}}^{-1}(\overleftarrow{v}%
\cdot v\overleftarrow{\mathbf{e}}_{\perp }),
\end{equation*}%
and $\ h\mathbf{\Theta =}\ ^{\alpha }\mathbf{D}_{h\mathbf{X}}^{-1}\left(
\overrightarrow{v}\otimes \overrightarrow{\varpi }-\overrightarrow{\varpi }%
\otimes \overrightarrow{v}\right) ,~v\mathbf{\Theta =}\ ^{\alpha }\mathbf{D}%
_{v\mathbf{X}}^{-1}\left( \overleftarrow{v}\otimes \overleftarrow{\varpi }-%
\overleftarrow{\varpi }\otimes \overleftarrow{v}\right) .$ Substituting
these values, respectively, in equations in (\ref{torseqd}), we express
\begin{eqnarray*}
\overrightarrow{\varpi }&=&-\ ^{\alpha }\mathbf{D}_{h\mathbf{X}}h%
\overrightarrow{\mathbf{e}}_{\perp }-\ ^{\alpha }\mathbf{D}_{h\mathbf{X}%
}^{-1}(\overrightarrow{v}\cdot h\overrightarrow{\mathbf{e}}_{\perp })%
\overrightarrow{v}, \\
\overleftarrow{\varpi }&=&-\ ^{\alpha }\mathbf{D}_{v\mathbf{X}}v%
\overleftarrow{\mathbf{e}}_{\perp }-\ ^{\alpha }\mathbf{D}_{v\mathbf{X}%
}^{-1}(\overleftarrow{v}\cdot v\overleftarrow{\mathbf{e}}_{\perp })%
\overleftarrow{v},
\end{eqnarray*}%
contained in the h-- and v--flow equations respectively on $\overrightarrow{v%
}$ and $\overleftarrow{v},$ considered as scalar components when $\ ^{\alpha
}\mathbf{D}_{h\mathbf{Y}}\overrightarrow{v}=\overrightarrow{v}_{\tau }$ and $%
\ ^{\alpha }\mathbf{D}_{h\mathbf{Y}}\overleftarrow{v}=\overleftarrow{v}%
_{\tau },$
\begin{eqnarray}
\overrightarrow{v}_{\tau } &=&\ ^{\alpha }\mathbf{D}_{h\mathbf{X}}%
\overrightarrow{\varpi }-\overrightarrow{v}\rfloor \ ^{\alpha }\mathbf{D}_{h%
\mathbf{X}}^{-1}\left( \overrightarrow{v}\otimes \overrightarrow{\varpi }-%
\overrightarrow{\varpi }\otimes \overrightarrow{v}\right) -\ ^{\alpha }%
\overrightarrow{R}h\overrightarrow{\mathbf{e}}_{\perp },  \label{floweq} \\
\overleftarrow{v}_{\tau } &=&\ ^{\alpha }\mathbf{D}_{v\mathbf{X}}%
\overleftarrow{\varpi }-\overleftarrow{v}\rfloor \ ^{\alpha }\mathbf{D}_{v%
\mathbf{X}}^{-1}\left( \overleftarrow{v}\otimes \overleftarrow{\varpi }-%
\overleftarrow{\varpi }\otimes \overleftarrow{v}\right) -\ ^{\alpha }%
\overleftarrow{S}v\overleftarrow{\mathbf{e}}_{\perp },  \notag
\end{eqnarray}%
where $\ ^{\alpha }\overrightarrow{R}$ and $\ ^{\alpha }\overleftarrow{S}$
are the scalar curvatures of chosen d--connection. For symmetric Riemannian
spaces like $SO(n+1)/SO(n)\simeq S^{n},$ the value $\overrightarrow{R}$ is
just the scalar curvature $\chi =1.$ On N--anholonomic fractional manifolds,
it is possible to define such d--connections that $\ ^{\alpha }%
\overrightarrow{R}$ and $\ ^{\alpha }\overleftarrow{S}$ are certain zero or
nonzero constants. $\square $
\end{proof}

The properties of operators (\ref{sop}) and (\ref{csop}) are defined by

\begin{theorem}
\label{mr1}The fractional d--operators $\ ^{\alpha }\mathcal{J=}\left( h\
^{\alpha }\mathcal{J},v\ ^{\alpha }\mathcal{J}\right) $ and\newline
$\ ^{\alpha }\mathcal{H=}\left( h\ ^{\alpha }\mathcal{H},v\ ^{\alpha }%
\mathcal{H}\right) $ $\ $are respectively $\left( O(n-1),O(m-1)\right) $%
--invariant Hamiltonian symplectic and cosymplectic d--operators with
respect to the fractional Hamiltonian d--variables $\left( \overrightarrow{v}%
,\overleftarrow{v}\right) .$ This class of d--operators defines the Hamiltonian
form for the curve fractional flow equations on
N--anholonomic fractional manifolds with constant
d--connection curvature: the fractional
h--flows are given by%
\begin{eqnarray}
\overrightarrow{v}_{\tau } &=&h\ ^{\alpha }\mathcal{H}\left( \overrightarrow{%
\varpi }\right) -\ ^{\alpha }\overrightarrow{R}~h\overrightarrow{\mathbf{e}}%
_{\perp }=h\ ^{\alpha }\mathfrak{R}\left( h\overrightarrow{\mathbf{e}}%
_{\perp }\right) -\ ^{\alpha }\overrightarrow{R}~h\overrightarrow{\mathbf{e}}%
_{\perp },  \notag \\
\overrightarrow{\varpi } &=&h\ ^{\alpha }\mathcal{J}\left( h\overrightarrow{%
\mathbf{e}}_{\perp }\right) ;  \label{hhfeq1}
\end{eqnarray}%
the fractional v--flows are given by
\begin{eqnarray}
\overleftarrow{v}_{\tau } &=&v\ ^{\alpha }\mathcal{H}\left( \overleftarrow{%
\varpi }\right) -\ ^{\alpha }\overleftarrow{S}~v\overleftarrow{\mathbf{e}}%
_{\perp }=v\ ^{\alpha }\mathfrak{R}\left( v\overleftarrow{\mathbf{e}}_{\perp
}\right) -\ ^{\alpha }\overleftarrow{S}~v\overleftarrow{\mathbf{e}}_{\perp },
\notag \\
\overleftarrow{\varpi } &=&v\ ^{\alpha }\mathcal{J}\left( v\overleftarrow{%
\mathbf{e}}_{\perp }\right) ,  \label{vhfeq1}
\end{eqnarray}%
where the so--called (fractional) heriditary recursion d--operator has the
respective h-- and v--components
\begin{equation}
h\ ^{\alpha }\mathfrak{R}=h\ ^{\alpha }\mathcal{H}\circ h\ ^{\alpha }%
\mathcal{J}\mbox{ \ and \ }v\ ^{\alpha }\mathfrak{R}=v\ ^{\alpha }\mathcal{H}%
\circ v\ ^{\alpha }\mathcal{J}.  \label{reqop}
\end{equation}
\end{theorem}

\begin{proof}
Such a proof follows from the Lemma and (\ref{floweq}). $\square $
\end{proof}

\section{Fractional Bi--Hamiltonians and Solitonic Hierarchies}

\label{s4} The fractional recursion h--operator from (\ref{reqop}),%
\begin{eqnarray}
h\ ^{\alpha }\mathfrak{R} &=&\ ^{\alpha }\mathbf{D}_{h\mathbf{X}}\left( \
^{\alpha }\mathbf{D}_{h\mathbf{X}}+\ ^{\alpha }\mathbf{D}_{h\mathbf{X}%
}^{-1}\left( \overrightarrow{v}\cdot \right) \overrightarrow{v}\right) +%
\overrightarrow{v}\rfloor \ ^{\alpha }\mathbf{D}_{h\mathbf{X}}^{-1}\left(
\overrightarrow{v}\wedge \ ^{\alpha }\mathbf{D}_{h\mathbf{X}}\right)  \notag
\\
&=&\ ^{\alpha }\mathbf{D}_{h\mathbf{X}}^{2}+|\ ^{\alpha }\mathbf{D}_{h%
\mathbf{X}}|^{2}+\ ^{\alpha }\mathbf{D}_{h\mathbf{X}}^{-1}\left(
\overrightarrow{v}\cdot \right) \overrightarrow{v}_{\mathbf{l}}-%
\overrightarrow{v}\rfloor \ ^{\alpha }\mathbf{D}_{h\mathbf{X}}^{-1}(%
\overrightarrow{v}_{\mathbf{l}}\wedge ),  \label{reqoph}
\end{eqnarray}%
induces a horizontal hierarchy of commuting Hamiltonian vector fields $h%
\overrightarrow{\mathbf{e}}_{\perp }^{(k)}$ starting from $h\overrightarrow{%
\mathbf{e}}_{\perp }^{(0)}=\overrightarrow{v}_{\mathbf{l}}.$ Such vector
fields are given by the infinitesimal generator of $\mathbf{l}$%
--translations in terms of arclength $\mathbf{l}$ along the curve. A
vertical hierarchy of commuting vector fields $v\overleftarrow{\mathbf{e}}%
_{\perp }^{(k)}$ starting from $v\overleftarrow{\mathbf{e}}_{\perp }^{(0)}$ $%
=\overleftarrow{v}_{\mathbf{l}}$ is generated by the recursion v--operator%
\begin{eqnarray}
v\ ^{\alpha }\mathfrak{R} &=&\ ^{\alpha }\mathbf{D}_{v\mathbf{X}}\left( \
^{\alpha }\mathbf{D}_{v\mathbf{X}}+\ ^{\alpha }\mathbf{D}_{v\mathbf{X}%
}^{-1}\left( \overleftarrow{v}\cdot \right) \overleftarrow{v}\right) +%
\overleftarrow{v}\rfloor \ ^{\alpha }\mathbf{D}_{v\mathbf{X}}^{-1}\left(
\overleftarrow{v}\wedge \ ^{\alpha }\mathbf{D}_{v\mathbf{X}}\right)  \notag
\\
&=&\ ^{\alpha }\mathbf{D}_{v\mathbf{X}}^{2}+|\ ^{\alpha }\mathbf{D}_{v%
\mathbf{X}}|^{2}+\ ^{\alpha }\mathbf{D}_{v\mathbf{X}}^{-1}\left(
\overleftarrow{v}\cdot \right) \overleftarrow{v}_{\mathbf{l}}-\overleftarrow{%
v}\rfloor \ ^{\alpha }\mathbf{D}_{v\mathbf{X}}^{-1}(\overleftarrow{v}_{%
\mathbf{l}}\wedge ).  \label{reqopv}
\end{eqnarray}%
We can associate fractional hierarchies generated by adjoint operators $\
^{\alpha }\mathfrak{R}^{\ast }=(h\ ^{\alpha }\mathfrak{R}^{\ast },$ $v\
^{\alpha }\mathfrak{R}^{\ast }),$ of involuntive fractional Hamiltonian
h--covector fields $\overrightarrow{\varpi }^{(k)}=\delta \left( h\ ^{\alpha
}H^{(k)}\right) /\delta \overrightarrow{v}$ in terms of fractional
Hamiltonians $h\ ^{\alpha }H=$\newline
$h\ ^{\alpha }H^{(k)}(\overrightarrow{v},\overrightarrow{v}_{\mathbf{l}},%
\overrightarrow{v}_{2\mathbf{l}},...)$ starting from $\overrightarrow{\varpi
}^{(0)}=\overrightarrow{v},h\ ^{\alpha }H^{(0)}=\frac{1}{2}|\overrightarrow{v%
}|^{2}$ and of involutive fractional Hamiltonian v--covector fields $%
\overleftarrow{\varpi }^{(k)}=\delta \left( v\ ^{\alpha }H^{(k)}\right) /$ $%
\delta \overleftarrow{v}$ in terms of Hamiltonians $v\ ^{\alpha }H=v\
^{\alpha }H^{(k)}(\overleftarrow{v},\overleftarrow{v}_{\mathbf{l}},%
\overleftarrow{v}_{2\mathbf{l}},...)$ starting from $\overleftarrow{\varpi }%
^{(0)}=\overleftarrow{v},v\ ^{\alpha }H^{(0)}=\frac{1}{2}|\overleftarrow{v}%
|^{2}.$ The relations between fractional hierarchies is given by formulas%
\begin{eqnarray*}
h\overrightarrow{\mathbf{e}}_{\perp }^{(k)} &=&h\ ^{\alpha }\mathcal{H}%
\left( \overrightarrow{\varpi }^{(k)},\overrightarrow{\varpi }%
^{(k+1)}\right) =h\ ^{\alpha }\mathcal{J}\left( h\overrightarrow{\mathbf{e}}%
_{\perp }^{(k)}\right) , \\
v\overleftarrow{\mathbf{e}}_{\perp }^{(k)} &=&v\ ^{\alpha }\mathcal{H}\left(
\overleftarrow{\varpi }^{(k)},\overleftarrow{\varpi }^{(k+1)}\right) =v\
^{\alpha }\mathcal{J}\left( v\overleftarrow{\mathbf{e}}_{\perp
}^{(k)}\right) ,
\end{eqnarray*}%
where $k=0,1,2,....$ All hierarchies (horizontal, vertical and their adjoint
ones) have a typical mKdV scaling symmetry, for instance, $\mathbf{%
l\rightarrow \lambda l}$ and $\overrightarrow{v}\rightarrow \mathbf{\lambda }%
^{-1}\overrightarrow{v}$ under which the values $h\overrightarrow{\mathbf{e}}%
_{\perp }^{(k)}$ and $h\ ^{\alpha }H^{(k)}$ have scaling weight $2+2k,$
while $\overrightarrow{\varpi }^{(k)}$ has scaling weight $1+2k.$

\begin{corollary}
\label{c2} There are N--adapted fractional hierarchies of distinguished
horizontal and vertical commuting bi--Hamiltonian fractional flows,
correspondingly, on $\overrightarrow{v}$ and $\overleftarrow{v},$ \
associated to the recursion d--operator (\ref{reqop}) given by $O(n-1)\oplus
O(m-1)$ --invariant d--vector evolution equations,%
\begin{eqnarray*}
\overrightarrow{v}_{\tau } &=&h\overrightarrow{\mathbf{e}}_{\perp
}^{(k+1)}-\ ^{\alpha }\overrightarrow{R}~h\overrightarrow{\mathbf{e}}_{\perp
}^{(k)}=h\ ^{\alpha }\mathcal{H}\left( \delta \left( h\ ^{\alpha }H^{(k,%
\overrightarrow{R})}\right) /\delta \overrightarrow{v}\right) \\
&=&\left( h\ ^{\alpha }\mathcal{J}\right) ^{-1}\left( \delta \left( h\
^{\alpha }H^{(k+1,\overrightarrow{R})}\right) /\delta \overrightarrow{v}%
\right)
\end{eqnarray*}%
with horizontal fractional Hamiltonians
\begin{equation*}
h\ ^{\alpha }H^{(k+1,\overrightarrow{R})}=h\ ^{\alpha }H^{(k+1,%
\overrightarrow{R})}-\ ^{\alpha }\overrightarrow{R}~h\ ^{\alpha }H^{(k,%
\overrightarrow{R})}
\end{equation*}
and
\begin{eqnarray*}
\overleftarrow{v}_{\tau } &=&v\overleftarrow{\mathbf{e}}_{\perp }^{(k+1)}-\
^{\alpha }\overleftarrow{S}~v\overleftarrow{\mathbf{e}}_{\perp }^{(k)}=v\
^{\alpha }\mathcal{H}\left( \delta \left( v\ ^{\alpha }H^{(k,\overleftarrow{S%
})}\right) /\delta \overleftarrow{v}\right) \\
&=&\left( v\ ^{\alpha }\mathcal{J}\right) ^{-1}\left( \delta \left( v\
^{\alpha }H^{(k+1,\overleftarrow{S})}\right) /\delta \overleftarrow{v}\right)
\end{eqnarray*}%
with vertical fractional Hamiltonians
\begin{equation*}
v\ ^{\alpha }H^{(k+1,\overleftarrow{S})}=v\ ^{\alpha }H^{(k+1,\overleftarrow{%
S})}-\ ^{\alpha }\overleftarrow{S}~v\ ^{\alpha }H^{(k,\overleftarrow{S})},
\end{equation*}
for $k=0,1,2,.....$ The fractional d--operators $\ ^{\alpha }\mathcal{H}$
and $\ ^{\alpha }\mathcal{J}$ $\ $are N--adapted and mutually compatible
from which one can be constructed an alternative (explicit) fractional
Hamilton d--operator $~^{a}\mathcal{H=\ ^{\alpha }H\circ \ ^{\alpha }J}$ $%
\circ \ ^{\alpha }\mathcal{H=}\ ^{\alpha }\mathfrak{R\circ }\ ^{\alpha }%
\mathcal{H}.$
\end{corollary}

\begin{proof}
It follows from above presented considerations. $\square $
\end{proof}

\subsection{Formulation of the Main Theorem}

Our goal is to prove that the geometric data for any fractional metric (in a
model of fractional gravity or geometric mechanics) naturally define a
N--adapted fractional bi--Hamiltonian flow hierarchy inducing anholonomic
fractional solitonic configurations.

\begin{theorem}
\label{mt} For any N--anholonomic fractional manifold with prescribed
fractional d--metric structure, there is a hierarchy of bi-Hamiltonian
N--adapted fractional flows of curves $\gamma (\tau ,\mathbf{l})=h\gamma
(\tau ,\mathbf{l})+v\gamma (\tau ,\mathbf{l})$ described by geometric
nonholonomic fractional map equations. The $0$ fractional flows are defined
as convective (traveling wave) maps%
\begin{equation}
\gamma _{\tau }=\gamma _{\mathbf{l}},\mbox{\ distinguished \ }\left( h\gamma
\right) _{\tau }=\left( h\gamma \right) _{h\mathbf{X}}\mbox{\ and \ }\left(
v\gamma \right) _{\tau }=\left( v\gamma \right) _{v\mathbf{X}}.
\label{trmap}
\end{equation}%
There are \ fractional +1 flows defined as non--stretching mKdV maps%
\begin{eqnarray}
-\left( h\gamma \right) _{\tau } &=&\ ^{\alpha }\mathbf{D}_{h\mathbf{X}%
}^{2}\left( h\gamma \right) _{h\mathbf{X}}+\frac{3}{2}\left| \ ^{\alpha }%
\mathbf{D}_{h\mathbf{X}}\left( h\gamma \right) _{h\mathbf{X}}\right| _{h%
\mathbf{g}}^{2}~\left( h\gamma \right) _{h\mathbf{X}},  \label{1map} \\
-\left( v\gamma \right) _{\tau } &=&\ ^{\alpha }\mathbf{D}_{v\mathbf{X}%
}^{2}\left( v\gamma \right) _{v\mathbf{X}}+\frac{3}{2}\left| \ ^{\alpha }%
\mathbf{D}_{v\mathbf{X}}\left( v\gamma \right) _{v\mathbf{X}}\right| _{v%
\mathbf{g}}^{2}~\left( v\gamma \right) _{v\mathbf{X}},  \notag
\end{eqnarray}%
and fractional +2,... flows as higher order analogs. Finally, the fractional
-1 flows are defined by the kernels of recursion fractional operators (\ref%
{reqoph}) and (\ref{reqopv}) inducing non--stretching fractional maps%
\begin{equation}
\ ^{\alpha }\mathbf{D}_{h\mathbf{Y}}\left( h\gamma \right) _{h\mathbf{X}}=0%
\mbox{\ and \
}\ ^{\alpha }\mathbf{D}_{v\mathbf{Y}}\left( v\gamma \right) _{v\mathbf{X}}=0.
\label{-1map}
\end{equation}
\end{theorem}

\begin{proof}
It is given below in section \ref{ssp}. $\square $
\end{proof}

\subsection{Proof of the Main Theorem}

\label{ssp}We generalize for fractional spaces a similar proof from Refs. %
\cite{vacap,vanco} sketching the key steps for horizontal flows. The
vertical constructions are similar but with respective changing of h--
variables / objects into v- variables/ objects. By corresponding
nonholonomic constraints we can emphasize certain h-- and v--evolutions
which are inter--related.

We get a fractional vector mKdV equation up to a convective term (which can
be absorbed by redefinition of coordinates) defining the \ fractional +1
flow for $h\overrightarrow{\mathbf{e}}_{\perp }=\overrightarrow{v}_{\mathbf{l%
}},$%
\begin{equation*}
\overrightarrow{v}_{\tau }=\overrightarrow{v}_{3\mathbf{l}}+\frac{3}{2}|%
\overrightarrow{v}|^{2}-\ ^{\alpha }\overrightarrow{R}~\overrightarrow{v}_{%
\mathbf{l}},
\end{equation*}%
when the fractional $+(k+1)$ flow gives a vector mKdV equation of higher
order $3+2k$ on $\overrightarrow{v}$ and there is a $0$ h--flow $%
\overrightarrow{v}_{\tau }=\overrightarrow{v}_{\mathbf{l}}$ arising from $h%
\overrightarrow{\mathbf{e}}_{\perp }=0$ and $h\overrightarrow{\mathbf{e}}%
_{\parallel }=1$ belonging outside the hierarchy generated by $h\ ^{\alpha }%
\mathfrak{R.}$ Such fractional flows correspond to N--adapted horizontal
motions of the curve $\gamma (\tau ,\mathbf{l})=h\gamma (\tau ,\mathbf{l}%
)+v\gamma (\tau ,\mathbf{l}),$ given by
\begin{equation*}
\left( h\gamma \right) _{\tau }=f\left( \left( h\gamma \right) _{h\mathbf{X}%
},\ ^{\alpha }\mathbf{D}_{h\mathbf{X}}\left( h\gamma \right) _{h\mathbf{X}%
},\ ^{\alpha }\mathbf{D}_{h\mathbf{X}}^{2}\left( h\gamma \right) _{h\mathbf{X%
}},...\right)
\end{equation*}%
subject to the non--stretching condition $|\left( h\gamma \right) _{h\mathbf{%
X}}|_{h\mathbf{g}}=1,$ when the equation of fractional motion is to be
derived from the identifications
\begin{equation*}
\left( h\gamma \right) _{\tau }\longleftrightarrow \ ^{\alpha }\mathbf{e}_{h%
\mathbf{Y}},\ ^{\alpha }\mathbf{D}_{h\mathbf{X}}\left( h\gamma \right) _{h%
\mathbf{X}}\longleftrightarrow \ ^{\alpha }\mathcal{D}_{h\mathbf{X}}\
^{\alpha }\mathbf{e}_{h\mathbf{X}}=\left[ \ ^{\alpha }\mathbf{L}_{h\mathbf{X}%
},\ ^{\alpha }\mathbf{e}_{h\mathbf{X}}\right]
\end{equation*}%
and so on, which maps the constructions from the tangent fractional space of
the curve to the space $h\mathfrak{p}.$ For such identifications, we have
\begin{eqnarray*}
\left[ \ ^{\alpha }\mathbf{L}_{h\mathbf{X}},\ ^{\alpha }\mathbf{e}_{h\mathbf{%
X}}\right] &=&-\left[
\begin{array}{cc}
0 & \left( 0,\overrightarrow{v}\right) \\
-\left( 0,\overrightarrow{v}\right) ^{T} & h\mathbf{0}%
\end{array}%
\right] \in h\mathfrak{p}, \\
\left[ \ ^{\alpha }\mathbf{L}_{h\mathbf{X}},\left[ \ ^{\alpha }\mathbf{L}_{h%
\mathbf{X}},\ ^{\alpha }\mathbf{e}_{h\mathbf{X}}\right] \right] &=&-\left[
\begin{array}{cc}
0 & \left( |\overrightarrow{v}|^{2},\overrightarrow{0}\right) \\
-\left( |\overrightarrow{v}|^{2},\overrightarrow{0}\right) ^{T} & h\mathbf{0}%
\end{array}%
\right]
\end{eqnarray*}%
and so on, see similar calculus in (\ref{aux41}). Stating for the \
fractional +1 h--flow
\begin{equation*}
h\overrightarrow{\mathbf{e}}_{\perp }=\overrightarrow{v}_{\mathbf{l}}%
\mbox{
and }h\overrightarrow{\mathbf{e}}_{\parallel }=-\ ^{\alpha }\mathbf{D}_{h%
\mathbf{X}}^{-1}\left( \overrightarrow{v}\cdot \overrightarrow{v}_{\mathbf{l}%
}\right) =-\frac{1}{2}|\overrightarrow{v}|^{2},
\end{equation*}%
we compute
\begin{eqnarray*}
\ ^{\alpha }\mathbf{e}_{h\mathbf{Y}} &=&\left[
\begin{array}{cc}
0 & \left( h\mathbf{e}_{\parallel },h\overrightarrow{\mathbf{e}}_{\perp
}\right) \\
-\left( h\mathbf{e}_{\parallel },h\overrightarrow{\mathbf{e}}_{\perp
}\right) ^{T} & h\mathbf{0}%
\end{array}%
\right] \\
&=&-\frac{1}{2}|\overrightarrow{v}|^{2}\left[
\begin{array}{cc}
0 & \left( 1,\overrightarrow{\mathbf{0}}\right) \\
-\left( 0,\overrightarrow{\mathbf{0}}\right) ^{T} & h\mathbf{0}%
\end{array}%
\right] +\left[
\begin{array}{cc}
0 & \left( 0,\overrightarrow{v}_{h\mathbf{X}}\right) \\
-\left( 0,\overrightarrow{v}_{h\mathbf{X}}\right) ^{T} & h\mathbf{0}%
\end{array}%
\right] \\
&=&\ ^{\alpha }\mathbf{D}_{h\mathbf{X}}\left[ \ ^{\alpha }\mathbf{L}_{h%
\mathbf{X}},\ ^{\alpha }\mathbf{e}_{h\mathbf{X}}\right] +\frac{1}{2}\left[ \
^{\alpha }\mathbf{L}_{h\mathbf{X}},\left[ \ ^{\alpha }\mathbf{L}_{h\mathbf{X}%
},\ ^{\alpha }\mathbf{e}_{h\mathbf{X}}\right] \right] \\
&=&-\ ^{\alpha }\mathcal{D}_{h\mathbf{X}}\left[ \ ^{\alpha }\mathbf{L}_{h%
\mathbf{X}},\ ^{\alpha }\mathbf{e}_{h\mathbf{X}}\right] -\frac{3}{2}|%
\overrightarrow{v}|^{2}\ ^{\alpha }\mathbf{e}_{h\mathbf{X}}.
\end{eqnarray*}%
The above identifications are related to the first and second terms, when
\begin{eqnarray*}
|\overrightarrow{v}|^{2} &=&<\left[ \ ^{\alpha }\mathbf{L}_{h\mathbf{X}},\
^{\alpha }\mathbf{e}_{h\mathbf{X}}\right] , \\
&& \left[ \ ^{\alpha }\mathbf{L}_{h\mathbf{X}},\ ^{\alpha }\mathbf{e}_{h%
\mathbf{X}}\right] >_{h\mathfrak{p}}\longleftrightarrow h\ ^{\alpha }\mathbf{%
g}\left( \ ^{\alpha }\mathbf{D}_{h\mathbf{X}}\left( h\gamma \right) _{h%
\mathbf{X}},\ ^{\alpha }\mathbf{D}_{h\mathbf{X}}\left( h\gamma \right) _{h%
\mathbf{X}}\right) \\
&=&\left| \ ^{\alpha }\mathbf{D}_{h\mathbf{X}}\left( h\gamma \right) _{h%
\mathbf{X}}\right| _{h\ ^{\alpha }\mathbf{g}}^{2},
\end{eqnarray*}%
and allow us to identify $\ ^{\alpha }\mathcal{D}_{h\mathbf{X}}\left[ \
^{\alpha }\mathbf{L}_{h\mathbf{X}},\ ^{\alpha }\mathbf{e}_{h\mathbf{X}}%
\right] $ to $\ ^{\alpha }\mathbf{D}_{h\mathbf{X}}^{2}\left( h\gamma \right)
_{h\mathbf{X}}.$ As a result, we have
\begin{equation*}
-\ ^{\alpha }\mathbf{e}_{h\mathbf{Y}}\longleftrightarrow \ ^{\alpha }\mathbf{%
D}_{h\mathbf{X}}^{2}\left( h\gamma \right) _{h\mathbf{X}}+\frac{3}{2}\left|
\ ^{\alpha }\mathbf{D}_{h\mathbf{X}}\left( h\gamma \right) _{h\mathbf{X}%
}\right| _{h\ ^{\alpha }\mathbf{g}}^{2}~\left( h\gamma \right) _{h\mathbf{X}}
\end{equation*}%
which is just the fractional equation (\ref{1map}) in the Theorem \ref{mt}
defining a non--stretching mKdV map h--equation induced by the h--part of
the fractional canonical d--connection.

To derive the higher \ order terms of hierarchies, we use the adjoint
representation $ad\left( \cdot \right) $ acting in the Lie algebra $h%
\mathfrak{g}=h\mathfrak{p}\oplus \mathfrak{so}(n),$ with
\begin{equation*}
ad\left( \left[ \ ^{\alpha }\mathbf{L}_{h\mathbf{X}},\ ^{\alpha }\mathbf{e}%
_{h\mathbf{X}}\right] \right) \ ^{\alpha }\mathbf{e}_{h\mathbf{X}}=\left[
\begin{array}{cc}
0 & \left( 0,\overrightarrow{\mathbf{0}}\right) \\
-\left( 0,\overrightarrow{\mathbf{0}}\right) ^{T} & \overrightarrow{\mathbf{v%
}}%
\end{array}%
\right] \in \mathfrak{so}(n+1),
\end{equation*}%
where $\overrightarrow{\mathbf{v}}=-\left[
\begin{array}{cc}
0 & \overrightarrow{v} \\
-\overrightarrow{v}^{T} & h\mathbf{0}%
\end{array}%
\in \mathfrak{so}(n)\right] .$ Applying again $\ ad\left( \left[ \ ^{\alpha }%
\mathbf{L}_{h\mathbf{X}},\ ^{\alpha }\mathbf{e}_{h\mathbf{X}}\right]
\right), $ we get
\begin{eqnarray*}
ad\left( \left[ \ ^{\alpha }\mathbf{L}_{h\mathbf{X}},\ ^{\alpha }\mathbf{e}%
_{h\mathbf{X}}\right] \right) ^{2}\ ^{\alpha }\mathbf{e}_{h\mathbf{X}}&=&-|%
\overrightarrow{v}|^{2}\left[
\begin{array}{cc}
0 & \left( 1,\overrightarrow{\mathbf{0}}\right) \\
-\left( 1,\overrightarrow{\mathbf{0}}\right) ^{T} & \mathbf{0}%
\end{array}%
\right] \\
&=&-|\overrightarrow{v}|^{2}\ ^{\alpha }\mathbf{e}_{h\mathbf{X}},
\end{eqnarray*}%
when the fractional equation (\ref{1map}) can be represented in alternative
form
\begin{equation*}
-\left( h\gamma \right) _{\tau }=\ ^{\ \alpha }\mathbf{D}_{h\mathbf{X}%
}^{2}\left( h\gamma \right) _{h\mathbf{X}}-\frac{3}{2}\ ^{\alpha }%
\overrightarrow{R}^{-1}ad\left( \ ^{\alpha }\mathbf{D}_{h\mathbf{X}}\left(
h\gamma \right) _{h\mathbf{X}}\right) ^{2}~\left( h\gamma \right) _{h\mathbf{%
X}}.
\end{equation*}

Finally, we consider a fractional -1 flow contained in the h--hierarchy
derived from the property that $h\overrightarrow{\mathbf{e}}_{\perp }$ is
annihilated by the fractional h--operator $h\ ^{\ \alpha }\mathcal{J}$ and
mapped into $h\ ^{\ \alpha }\mathfrak{R}(h\overrightarrow{\mathbf{e}}_{\perp
})=0.$ This mean that $h\ ^{\ \alpha }\mathcal{J}(h\overrightarrow{\mathbf{e}%
}_{\perp })=\overrightarrow{\varpi }=0.$ Such properties together with (\ref%
{auxaaa}) and fractional equations (\ref{floweq}) imply $\ ^{\ \alpha }%
\mathbf{L}_{\tau }=0$ and hence
\begin{equation*}
h\ ^{\ \alpha }\mathcal{D}_{\tau }\ ^{\ \alpha }\mathbf{e}_{h\mathbf{X}}=[\
^{\ \alpha }\mathbf{L}_{\tau },\ ^{\ \alpha }\mathbf{e}_{h\mathbf{X}}]=0 %
\mbox{ for } h\ ^{\ \alpha }\mathcal{D}_{\tau }=h\ ^{\ \alpha }\mathbf{D}%
_{\tau }+[\ ^{\ \alpha }\mathbf{L}_{\tau },\cdot ].
\end{equation*}
We obtain the equation of fractional motion for the h--component of curve, $%
h\gamma (\tau ,\mathbf{l}),$ following the correspondences $\ ^{\ \alpha }%
\mathbf{D}_{h\mathbf{Y}}\longleftrightarrow h\ ^{\ \alpha }\mathcal{D}_{\tau
}$ and $h\gamma _{\mathbf{l}}\longleftrightarrow \ ^{\ \alpha }\mathbf{e}_{h%
\mathbf{X}},$ $\ ^{\ \alpha }\mathbf{D}_{h\mathbf{Y}}\left( h\gamma (\tau ,%
\mathbf{l})\right) =0,$ which is just the first fractional equation in (\ref%
{-1map}).

\section{Fractional Nonholonomic mKdV and SG Hierarchies}

\label{s5} In this section, we consider explicit constructions when
fractional solitonic hierarchies are derived following the conditions of
Theorem \ref{mt}.

The fractional h--flow and v--flow equations resulting from (\ref{-1map}) are%
\begin{equation}
\overrightarrow{v}_{\tau }=-\ ^{\ \alpha }\overrightarrow{R}h\overrightarrow{%
\mathbf{e}}_{\perp }\mbox{ \ and \ }\overleftarrow{v}_{\tau }=-\ ^{\ \alpha }%
\overleftarrow{S}v\overleftarrow{\mathbf{e}}_{\perp },  \label{deveq}
\end{equation}%
when, respectively,%
\begin{equation*}
0=\overrightarrow{\varpi }=-\ ^{\ \alpha }\mathbf{D}_{h\mathbf{X}}h%
\overrightarrow{\mathbf{e}}_{\perp }+h\mathbf{e}_{\parallel }\overrightarrow{%
v},~\ ^{\ \alpha }\mathbf{D}_{h\mathbf{X}}h\mathbf{e}_{\parallel }=h%
\overrightarrow{\mathbf{e}}_{\perp }\cdot \overrightarrow{v}
\end{equation*}%
and
\begin{equation*}
0=\overleftarrow{\varpi }=-\ ^{\ \alpha }\mathbf{D}_{v\mathbf{X}}v%
\overleftarrow{\mathbf{e}}_{\perp }+v\mathbf{e}_{\parallel }\overleftarrow{v}%
,~\ ^{\ \alpha }\mathbf{D}_{v\mathbf{X}}v\mathbf{e}_{\parallel }=v%
\overleftarrow{\mathbf{e}}_{\perp }\cdot \overleftarrow{v}.
\end{equation*}%
The fractional d--flow equations possess horizontal and vertical
conservation laws%
\begin{equation*}
\ ^{\ \alpha }\mathbf{D}_{h\mathbf{X}}\left( (h\mathbf{e}_{\parallel
})^{2}+|h\overrightarrow{\mathbf{e}}_{\perp }|^{2}\right) =0,
\end{equation*}%
for $(h\mathbf{e}_{\parallel })^{2}+|h\overrightarrow{\mathbf{e}}_{\perp
}|^{2}=<h\mathbf{e}_{\tau },h\mathbf{e}_{\tau }>_{h\mathfrak{p}}=|\left(
h\gamma \right) _{\tau }|_{h\ ^{\ \alpha }\mathbf{g}}^{2},$ and
\begin{equation*}
\ ^{\ \alpha }\mathbf{D}_{v\mathbf{Y}}\left( (v\mathbf{e}_{\parallel
})^{2}+|v\overleftarrow{\mathbf{e}}_{\perp }|^{2}\right) =0,
\end{equation*}%
for $(v\mathbf{e}_{\parallel })^{2}+|v\overleftarrow{\mathbf{e}}_{\perp
}|^{2}=<v\mathbf{e}_{\tau },v\mathbf{e}_{\tau }>_{v\mathfrak{p}}=|\left(
v\gamma \right) _{\tau }|_{v\ ^{\ \alpha }\mathbf{g}}^{2}.$ This corresponds
to
\begin{equation*}
\ ^{\ \alpha }\mathbf{D}_{h\mathbf{X}}|\left( h\gamma \right) _{\tau }|_{h\
^{\ \alpha }\mathbf{g}}^{2}=0\mbox{ \ and \ }\ ^{\ \alpha }\mathbf{D}_{v%
\mathbf{X}}|\left( v\gamma \right) _{\tau }|_{v\ ^{\ \alpha }\mathbf{g}%
}^{2}=0.
\end{equation*}

We can also rescale conformally the variable $\tau $ in order to get $%
|\left( h\gamma \right) _{\tau }|_{h\ ^{\ \alpha }\mathbf{g}}^{2}$ $=1$ and
(it could be for other rescaling) $|\left( v\gamma \right) _{\tau }|_{v\ ^{\
\alpha }\mathbf{g}}^{2}=1,$ i.e. to have%
\begin{equation*}
(h\mathbf{e}_{\parallel })^{2}+|h\overrightarrow{\mathbf{e}}_{\perp }|^{2}=1%
\mbox{ \ and \ }(v\mathbf{e}_{\parallel })^{2}+|v\overleftarrow{\mathbf{e}}%
_{\perp }|^{2}=1.
\end{equation*}%
Then, we express $h\mathbf{e}_{\parallel }$ and $h\overrightarrow{\mathbf{e}}%
_{\perp }$ in terms of $\overrightarrow{v}$ and its derivatives and,
similarly, we express $v\mathbf{e}_{\parallel }$ and $v\overleftarrow{%
\mathbf{e}}_{\perp }$ in terms of $\overleftarrow{v}$ and its derivatives,
which follows from (\ref{deveq}). The N--adapted fractional wave map
equations describing the -1 flows reduce to a system of two independent
nonlocal fractional evolution equations for the h-- and v--components,%
\begin{eqnarray*}
\overrightarrow{v}_{\tau }&=&-\ ^{\ \alpha }\mathbf{D}_{h\mathbf{X}%
}^{-1}\left( \sqrt{\ ^{\ \alpha }\overrightarrow{R}^{2}-|\overrightarrow{v}%
_{\tau }|^{2}}~\overrightarrow{v}\right), \\
\overleftarrow{v}_{\tau }&=&-\ ^{\ \alpha }\mathbf{D}_{v\mathbf{X}%
}^{-1}\left( \sqrt{\ ^{\ \alpha }\overleftarrow{S}^{2}-|\overleftarrow{v}%
_{\tau }|^{2}}~\overleftarrow{v}\right) .
\end{eqnarray*}
We can rescale the equations on $\tau $ to the case when the terms $\ ^{\
\alpha }\overrightarrow{R}^{2},\ ^{\ \alpha }\overleftarrow{S}^{2}=1,$ and
the fractional evolution equations transform into a system of hyperbolic
d--vector equations,%
\begin{equation}
\ ^{\ \alpha }\mathbf{D}_{h\mathbf{X}}(\overrightarrow{v}_{\tau })=-\sqrt{1-|%
\overrightarrow{v}_{\tau }|^{2}}~\overrightarrow{v}\mbox{ \ and \ }\ ^{\
\alpha }\mathbf{D}_{v\mathbf{X}}(\overleftarrow{v}_{\tau })=-\sqrt{1-|%
\overleftarrow{v}_{\tau }|^{2}}~\overleftarrow{v},  \label{heq}
\end{equation}%
where $\ ^{\ \alpha }\mathbf{D}_{h\mathbf{X}}=\ ^{\ \alpha }\partial _{h%
\mathbf{l}}$ and $\ ^{\ \alpha }\mathbf{D}_{v\mathbf{X}}=\ ^{\ \alpha
}\partial _{v\mathbf{l}}$ are usual partial derivatives on direction $%
\mathbf{l=}h\mathbf{l+}v\mathbf{l}$ with $\overrightarrow{v}_{\tau }$ and $%
\overleftarrow{v}_{\tau }$ considered as scalar functions for the covariant
derivatives $\ ^{\ \alpha }\mathbf{D}_{h\mathbf{X}}$ and $\ ^{\ \alpha }%
\mathbf{D}_{v\mathbf{X}}$ defined by the fractional canonical d--connection.
It also follows that $h\overrightarrow{\mathbf{e}}_{\perp }$ and $v%
\overleftarrow{\mathbf{e}}_{\perp }$ obey corresponding fractional vector
sine--Gordon (SG) equations%
\begin{eqnarray}
\left( \sqrt{(1-|h\overrightarrow{\mathbf{e}}_{\perp }|^{2})^{-1}}~\ ^{\
\alpha }\partial _{h\mathbf{l}}(h\overrightarrow{\mathbf{e}}_{\perp
})\right) _{\tau }&=&-h\overrightarrow{\mathbf{e}}_{\perp }  \label{sgeh} \\
\left( \sqrt{(1-|v\overleftarrow{\mathbf{e}}_{\perp }|^{2})^{-1}}~\ ^{\
\alpha }\partial _{v\mathbf{l}}(v\overleftarrow{\mathbf{e}}_{\perp })\right)
_{\tau }&=&-v\overleftarrow{\mathbf{e}}_{\perp }.  \label{sgev}
\end{eqnarray}

\begin{conclusion}
The recursion fractional d--operator $\ ^{\ \alpha }\mathfrak{R}=(h\ ^{\
\alpha }\mathfrak{R,}h\ ^{\ \alpha }\mathfrak{R})$ (\ref{reqop}), see (\ref%
{reqoph}) and (\ref{reqopv}), generates two hierarchies of fractional vector
mKdV symmetries: the first one is horizontal,
\begin{eqnarray}
\overrightarrow{v}_{\tau }^{(0)} &=&\overrightarrow{v}_{h\mathbf{l}},~%
\overrightarrow{v}_{\tau }^{(1)}=h\ ^{\ \alpha }\mathfrak{R}(\overrightarrow{%
v}_{h\mathbf{l}})=\overrightarrow{v}_{3h\mathbf{l}}+\frac{3}{2}|%
\overrightarrow{v}|^{2}~\overrightarrow{v}_{h\mathbf{l}},  \label{mkdv1a} \\
\overrightarrow{v}_{\tau }^{(2)} &=&h\ ^{\ \alpha }\mathfrak{R}^{2}(%
\overrightarrow{v}_{h\mathbf{l}})=\overrightarrow{v}_{5h\mathbf{l}}+\frac{5}{%
2}\left( |\overrightarrow{v}|^{2}~\overrightarrow{v}_{2h\mathbf{l}}\right)
_{h\mathbf{l}}  \notag \\
&&+\frac{5}{2}\left( (|\overrightarrow{v}|^{2})_{h\mathbf{l~}h\mathbf{l}}+|%
\overrightarrow{v}_{h\mathbf{l}}|^{2}+\frac{3}{4}|\overrightarrow{v}%
|^{4}\right) ~\overrightarrow{v}_{h\mathbf{l}}-\frac{1}{2}|\overrightarrow{v}%
_{h\mathbf{l}}|^{2}~\overrightarrow{v},  \notag \\
&&...,  \notag
\end{eqnarray}%
with all such terms commuting with the fractional -1 flow
\begin{equation}
(\overrightarrow{v}_{\tau })^{-1}=h\overrightarrow{\mathbf{e}}_{\perp }
\label{mkdv1b}
\end{equation}%
associated to the fractional vector SG equation (\ref{sgeh}); the second one
is vertical,
\begin{eqnarray}
\overleftarrow{v}_{\tau }^{(0)} &=&\overleftarrow{v}_{v\mathbf{l}},~%
\overleftarrow{v}_{\tau }^{(1)}=v\ ^{\ \alpha }\mathfrak{R}(\overleftarrow{v}%
_{v\mathbf{l}})=\overleftarrow{v}_{3v\mathbf{l}}+\frac{3}{2}|\overleftarrow{v%
}|^{2}~\overleftarrow{v}_{v\mathbf{l}},  \label{mkdv2a} \\
\overleftarrow{v}_{\tau }^{(2)} &=&v\ ^{\ \alpha }\mathfrak{R}^{2}(%
\overleftarrow{v}_{v\mathbf{l}})=\overleftarrow{v}_{5v\mathbf{l}}+\frac{5}{2}%
\left( |\overleftarrow{v}|^{2}~\overleftarrow{v}_{2v\mathbf{l}}\right) _{v%
\mathbf{l}}  \notag \\
&&+\frac{5}{2}\left( (|\overleftarrow{v}|^{2})_{v\mathbf{l~}v\mathbf{l}}+|%
\overleftarrow{v}_{v\mathbf{l}}|^{2}+\frac{3}{4}|\overleftarrow{v}%
|^{4}\right) ~\overleftarrow{v}_{v\mathbf{l}}-\frac{1}{2}|\overleftarrow{v}%
_{v\mathbf{l}}|^{2}~\overleftarrow{v},  \notag \\
&&...,  \notag
\end{eqnarray}%
with all such terms commuting with the fractional -1 flow
\begin{equation}
(\overleftarrow{v}_{\tau })^{-1}=v\overleftarrow{\mathbf{e}}_{\perp }
\label{mkdv2b}
\end{equation}%
associated to the fractional vector SG equation (\ref{sgev}).
\end{conclusion}

\begin{proof}
It follows from the above, \ in this section, and Corollary \ref{c2}. $%
\square $
\end{proof}

Finally, using the above Conclusion, we derive that the adjoint fractional
d--operator $\ ^{\ \alpha }\mathfrak{R}^{\ast }=\ ^{\ \alpha }\mathcal{%
J\circ \ ^{\ \alpha }H}$ generates a horizontal hierarchy of fractional
Hamiltonians,%
\begin{eqnarray}
h\ ^{\ \alpha }H^{(0)} &=&\frac{1}{2}|\overrightarrow{v}|^{2},~h\ ^{\ \alpha
}H^{(1)}=-\frac{1}{2}|\overrightarrow{v}_{h\mathbf{l}}|^{2}+\frac{1}{8}|%
\overrightarrow{v}|^{4},  \label{hhh} \\
h\ ^{\ \alpha }H^{(2)} &=&\frac{1}{2}|\overrightarrow{v}_{2h\mathbf{l}}|^{2}-%
\frac{3}{4}|\overrightarrow{v}|^{2}~|\overrightarrow{v}_{h\mathbf{l}}|^{2}-%
\frac{1}{2}\left( \overrightarrow{v}\cdot \overrightarrow{v}_{h\mathbf{l}%
}\right) +\frac{1}{16}|\overrightarrow{v}|^{6},...,  \notag
\end{eqnarray}%
and vertical hierarchy of fractional Hamiltonians%
\begin{eqnarray}
v\ ^{\ \alpha }H^{(0)} &=&\frac{1}{2}|\overleftarrow{v}|^{2},~v\ ^{\ \alpha
}H^{(1)}=-\frac{1}{2}|\overleftarrow{v}_{v\mathbf{l}}|^{2}+\frac{1}{8}|%
\overleftarrow{v}|^{4},  \label{hhv} \\
v\ ^{\ \alpha }H^{(2)} &=&\frac{1}{2}|\overleftarrow{v}_{2v\mathbf{l}}|^{2}-%
\frac{3}{4}|\overleftarrow{v}|^{2}~|\overleftarrow{v}_{v\mathbf{l}}|^{2}-%
\frac{1}{2}\left( \overleftarrow{v}\cdot \overleftarrow{v}_{v\mathbf{l}%
}\right) +\frac{1}{16}|\overleftarrow{v}|^{6},...,  \notag
\end{eqnarray}%
all of which are conserved densities for respective horizontal and vertical
fractional -1 flows and determining higher conservation laws for the
corresponding hyperbolic fractional equations (\ref{sgeh}) and (\ref{sgev}).

Finally we note that two concrete examples of fractional solitonic solutions in gravity are studied in section 5.3 of Ref. \cite{vrfrg} . The "integer" solitonic hierarchies from \cite{vacap,vanco}   can be similarly generated in  "fractional" forms by using corresponding fractional Caputo derivative operators. The length of this paper does not allow us consider such applications.

\vskip5pt \textbf{Acknowledgement: } This work contains the results of a talk at the 3d Conference on "Nonlinear Science and Complexity", 28--31 July, 2010, \c Chankaya University, Ankara, Turkey.

\end{document}